\documentclass{article}

\usepackage{amsmath}
\usepackage{amsfonts}
\usepackage{amssymb}
\usepackage{graphicx}
\usepackage{tikz}
\usepackage{esvect}
\usepackage[ruled]{algorithm2e}
\usepackage{caption}
\usepackage{subcaption}
\usepackage[left= 3 cm,right=3 cm,top=3 cm,bottom=3 cm]{geometry}
\usepackage{amsthm}
\usepackage{authblk}

\usetikzlibrary{arrows}

\DeclareMathOperator*{\argmax}{argmax}

\newtheorem{theorem}{Theorem}
\newtheorem{definition}{Definition}
\newtheorem{lemma}{Lemma}

\newtheorem{corollary}{Corollary}

\newtheorem*{mainresult}{Main result}

\newcommand{\set}[1]{\left\{ #1 \right\}}
\newcommand{\card}[1]{\left| #1 \right|}

\newcommand{\checkperp}{\mbox{\textsf{check}}^{\perp}}
\newcommand{\ifend}{\textbf{endif}}

\newcommand{\diam}{\mbox{\textsf{diam}}}
\newcommand{\rad}{\mbox{\textsf{rad}}}
\newcommand{\ecc}{\mbox{\textsf{ecc}}}
\newcommand{\opp}{\mbox{\textsf{op}}}

\newcommand{\tsw}{2-\textsc{sweep}}
\newcommand{\fsw}{4-\textsc{sweep}}

\title{Diameter, radius and all eccentricities in linear time for constant-dimension median graphs\footnote{An extended abstract of this paper will appear at LAGOS 2021. This is the full version.}}
\author{Pierre Berg\'e\footnote{Email : berge@irif.fr} }
\author{Michel Habib}
\affil{IRIF, CNRS, Universit\'e de Paris, France}

\begin{document}

\maketitle

\abstract{
Median graphs form the class of graphs which is the most studied in metric graph theory. Recently, B\'en\'eteau {\em et al.} [2019] designed a linear-time algorithm computing both the $\Theta$-classes and the median set of median graphs. A natural question emerges: is there a linear-time algorithm computing the diameter and the radius for median graphs?

We answer positively to this question for median graphs $G$ with constant dimension $d$, {\em i.e.} the dimension of the largest induced hypercube of $G$. We propose a combinatorial algorithm computing all eccentricities of median graphs with running time $O(2^{O(d\log d)}n)$. As a consequence, this provides us with a linear-time algorithm determining both the diameter and the radius of median graphs with $d = O(1)$, such as cube-free median graphs. As the hypercube of dimension 4 is not planar, it shows also that all eccentricities of planar median graphs can be computed in $O(n)$.
}

\section{Introduction} \label{sec:intro}

We study one of the most fundamental problems in algorithmic graph theory related to distances: the \textit{diameter} and the \textit{radius}. Given an undirected graph $G=(V,E)$, the diameter is the maximum distance $d(u,v)$, $u,v \in V$, where $d(u,v)$ is the length of the shortest $(u,v)$-path. The eccentricity of a vertex $v$ is the maximum length of a shortest path starting from $v$: the diameter is thus the maximum eccentricity and the radius is defined as the minimum eccentricity. Both the diameter and the radius are basic parameters used to apprehend the structure of a graph.

Multiple \textit{Breadth First Search} (BFS)  suffice as a naive algorithm  to compute the distances between all pairs of vertices and, therefore, to obtain the diameter, the radius and all eccentricities  in $O(n\card{E})$ for a $n$-vertex graph. Unfortunately, no known algorithm is able to determine the diameter much faster than all distances. Furthermore, it was shown~\cite{AbGrWi15} that both computing the radius and all distances are equivalent problems under subcubic reductions.
Very efficient algorithms have been proposed for the diameter on certain classes of graphs, for example~\cite{AbWiWa16,Ca17,DuHaVi20}. Many works have also been devoted to approximation algorithms for this parameter. In particular, Chechik {\em et al.}~\cite{ChLaRoScTaWi14} showed that the diameter can be approximated within a factor $\frac{3}{2}$ in time $O^*(m^{\frac{3}{2}})$. On sparse graphs, a natural question is whether we can compute exactly the diameter in subquadratic time. It was shown~\cite{RoWi13} that no $O(n^{2-\varepsilon})$-time algorithm can achieve an approximation factor smaller than $\frac{3}{2}$ for the diameter on sparse graphs unless the Strong Exponential Time Hypothesis (SETH) fails.

Linear-time heuristics have been proposed to estimate the diameter in general graphs. The \tsw\ algorithm~\cite{MaLaHa08} consists in taking an arbitrary vertex $r_1$ of the graph, computing the farthest-to-$r_1$ vertex $a_1$ with a BFS and eventually finding the farthest-to-$a_1$ vertex $b_1$ with another BFS. The distance $d(a_1,b_1)$ is returned. This method works well on some special classes of graphs. It was shown that \tsw\ returns exactly the diameter on trees~\cite{Ha73}. Moreover, it approximates it with an additive error 1 on chordal graphs and an additive error 2 on AT-free graphs~\cite{CoDrKo03}. The \fsw\ algorithm~\cite{BoCrHaKoMaTa15} picks a vertex in the middle of a shortest $(a_1,b_1)$-path and computes another \tsw . It performs well on real-world undirected graphs~\cite{BoCrHaKoMaTa15}.

In this paper, we propose a linear time algorithm computing the diameter, the radius, and all eccentricities of constant-dimension median graphs. Median graphs are the graphs such that any triplet of vertices has a unique median. Put formally, given $x,y,z \in V$, there is a unique vertex $m(x,y,z)$ lying at the same time on some shortest $(x,y)$, $(y,z)$, and $(z,x)$-paths. Said differently, $m(x,y,z)$ is the unique vertex being metrically between $x$ and $y$, $y$ and $z$, $z$ and $x$. Median graphs are \textit{partial cubes}, {\em i.e.} isometric subgraphs of hypercubes. However, partial cubes are not necessarily median. The dimension $d$ of a median graph is the dimension of its largest induced hypercube. This parameter is upper-bounded by $\lfloor \log n \rfloor$.

Median graphs are related to numerous areas: universal algebra~\cite{Av61,BiKi47}, CAT(0) cube complexes~\cite{BaCh08,Ch00}, abstract models of concurrency~\cite{BaCo93,SaNiWi93}, and genetics~\cite{BaQuSaMa02,BaFoSyRi95}. They have strong structural properties and admit many characterizations, such as the Mulder's convex expansion~\cite{Mu78,Mu80}, and are related to hypercubes retracts~\cite{Ba84}, Cartesian products and gated amalgams~\cite{BaCh08}, but also Helly hypergraphs~\cite{MuSc79}. Median graphs are bipartite and have at most $dn \le n\log n$ edges. They do not contain induced $K_{2,3}$, otherwise a triplet of vertices would admit at least two medians. The cardinality function of hypercubes $\alpha_i(G)$ - the number of induced hypercubes of dimension $0\le i\le d$ in the median graph $G$ - verifies nice formulas~\cite{BrKlSk06,Ko09}. A key concept to understand the structure of median graphs is the equivalence relation $\Theta$, which is the reflexive and transitive closure of relation $\Theta_0$, where two edges are in $\Theta_0$ if they are opposite edges of a common 4-cycle. An equivalence class of $\Theta$ is called a $\Theta$-class. Each $\Theta$-class of a median graph form a matching cutset, splitting the graph into two connected components, called halfspaces. These halfspaces are convex. The number $q$ of $\Theta$-classes is less than $n$ and, moreover, it is exactly the dimension of the hypercube in which the median graph $G$ isometrically embeds. It was shown in~\cite{KlMuSk98} that value $q$ satisfies the Euler-type formula $2n-m-q\le 2$. The $\Theta$-classes can be identified in linear time $O(\card{E})=O(dn)$, using a Lexicographic BFS~\cite{BeChChVa20}.

To the best of our knowledge, there is no subquadratic algorithm for the diameter on median graphs. B\'en\'eteau {\em et al.}~\cite{BeChChVa20} and Ducoffe~\cite{Du20} recently formulated this open question. There exist efficient algorithms for other metric parameters on median graphs. Thanks to the linear time computation of $\Theta$-classes, the median set and the Wiener index can be determined in $O(\card{E})$~\cite{BeChChVa20}. The median set, {\em i.e.} the vertices $u$ of $G$ which minimize $\sum_{v \in V} d(u,v)$, satisfies the majority rule~\cite{BaBrChKlKoSu10,BaBa84}. Another challenging question is the recognition of median graphs. Two very efficient algorithms have been proposed. Using convex characterizations of halfspaces, Hagauer {\em et al.}~\cite{HaImKl99} showed that median graphs can be recognized in $O(n^{\frac{3}{2}}\sqrt{n})$. In~\cite{ImKlMu99}, a bijection between median graphs and triangle-free graphs is identified. As a consequence, the recognition times of these two families are the same, if we neglect poly-logarithmic factors. Taking a very efficient algorithm which detects triangle-free graphs~\cite{AlYuZw97} produces a recognition of median graphs running in $O((n\log^2 n)^{1.41})$.

Subfamilies of median graphs have also been studied in the literature. There is an algorithm computing the diameter and the radius in linear time for squaregraphs, which are planar cube-free median graphs where all inner vertices have degree at least four~\cite{BaChEp10,ChDrVa02}. Recently, distance and routing labeling schemes~\cite{ChLaRa19} of size $O(\log^3 n)$ were designed for cube-free median graphs, {\em i.e.} median graphs verifying $d\le 2$.

In this paper, we construct an algorithm computing the diameter for median graphs which runs in time $O(n)$ when $d=O(1)$. Concretely, our algorithm returns the exact value of the diameter for all median graphs $G$ in time $O(f(d)n)$, where $f(d)$ is an exponential fonction of $d$. It can be naturally extended to provide us with all eccentricities of the median graph with the same running time. Our contribution is summarized below.

\begin{mainresult}
There is a combinatorial algorithm which computes the diameter, the radius and all eccentricities in time $O^*(2^{d(\log (d)+1)}n)$ on median graphs.
\label{th:main_result}
\end{mainresult}

Notation $O^*$ neglects polynomials of $d$, which are also poly-logarithmic factors of $n$. A consequence of our main result is that, for any $d=O(1)$, all eccentricities can be determined in linear time $O(n)$ on $Q_d$-free median graphs. For example, this is the case for cube-free median graphs ($d\le 2$). Moreover, as $Q_4$ is not planar, planar median graphs are $Q_4$-free, so our algorithm is linear for this family of graphs. Obviously, such an algorithm is not linear for all median graphs, as parameter $d$ can attain value $\log n$.

Given that \tsw\ returns in linear time the diameter on trees, which are the median graphs of dimension $d=1$, one can wonder whether the heuristics \tsw\ and \fsw\ can compute the diameter of median graphs for larger values of $d$. Another contribution in this paper is the proof that it is not the case. We provide median graphs of dimension 2 on which \tsw\ and \fsw\ fail to determine the diameter.

The organization of the paper follows. In Section~\ref{sec:preliminaries}, we introduce the definitions and notions related to median graphs. We describe in detail the properties of $\Theta$-classes, in particular the orthogonality, which is a key concept in this work. In Section~\ref{sec:sweep}, we show that \tsw\ and \fsw\ do not allow us to determine the exact value of the diameter on median graphs. We provide instances on which they are unsuccessful. In Section~\ref{sec:diameter}, we present a first version of our algorithm which computes the diameter only. We proceed in two steps: the computation of labels for the hypercubes of $G$ (Section~\ref{subsec:labels}) and the reduction to an optimization problem we call \textsc{maximum-weighted disjoint sets} (Section~\ref{subsec:mwds}). In Section~\ref{sec:eccentricities}, we introduce an additional step in our algorithm to determine the eccentricity of each vertex of $G$. We conclude this work in Section~\ref{sec:conclusion} by listing possible lines of future research on median graphs.

\section{Preliminaries} \label{sec:preliminaries}

In this section, we recall some basic notions related to graph theory, distances, and more particularly median graphs. The results listed are either reminded from the literature or direct consequences of earlier works.

\subsection{Graphs and distances} 

All graphs $G = (V,E)$ considered in this paper are undirected, unweighted, simple, finite and connected. Edges $(u,v) \in E$ are sometimes denoted by $uv$ to improve the readability of our article. Given two vertices $u,v \in V$, we denote by $d(u,v)$ the \textit{distance} between $u$ and $v$, {\em i.e.} the length of the shortest $(u,v)$-path. The interval $I(u,v)$ contains exactly the vertices which are on shortest $(u,v)$-paths. Put formally,
\[
I(u,v) = \set{x \in V: d(u,x) + d(x,v) = d(u,v)}
\]

We denote by $\ecc(u)$ the \textit{eccentricity} of a vertex $u \in V$, {\em i.e.} the maximum distance $d(u,v)$ for all $v \in V$: $\ecc(u) = \max_{v \in V} d(u,v)$.
The diameter $\diam(G)$ of graph $G$ is its maximum eccentricity: $\diam(G) = \max_{u \in V} \ecc(u)$. Concretely, the diameter is the length of the longest shortest path of $G$. The radius $\rad(G)$ is the minimum eccentricity: $\rad(G) = \min_{u \in V} \ecc(u)$.

Let $H\subseteq V$ be a vertex set. We say that $H$ (or the induced subgraph $G\left[H\right]$) is \textit{convex} if $I(u,v) \subseteq H$ for any pair $u,v \in H$. We say that $H$ is \textit{gated} if any vertex $v \notin H$ admits a gate $v' \in H$, {\em i.e.} a vertex that belongs to all intervals $I(v,x)$, $x\in H$. In other words, for any $x \in H$, we have $d(v,v') + d(v',x) = d(v,x)$. 

Given an integer $k \ge 1$, the $k$-dimensional hypercube $Q_k$ represents all the subsets of $\set{1,\ldots,k}$ as the vertex set. Two subsets $A$ and $B$ are connected by an edge if they differ only by one element, {\em i.e.} $\card{A \bigtriangleup B} = 1$. The hypercube $Q_2$ is called a \textit{square}.

\subsection{Median graphs} 

A graph is \textit{median} if, for any triplet $x,y,z$ of distinct vertices, the set $I(x,y) \cap I(y,z) \cap I(z,x)$ contains exactly one vertex called the median $m(x,y,z)$. A median graph is bipartite, triangle-free, and does not contain an induced $K_{2,3}$~\cite{BaCh08,HaImKl11,Mu78}. Some well-known families of graphs are median: trees, grids, squaregraphs~\cite{BaChEp10}, and hypercubes. The dimension $d = \mbox{dim}(G)$ of a median graph $G$ is the dimension of the largest hypercube contained in $G$ as an induced subgraph.


Figure~\ref{fig:median_examples} presents three examples of median graphs. One has dimension $d=1$ and is (necessarily) a tree. One has dimension $d=2$, a cogwheel. Cogwheels are squaregraphs but certain median graphs of dimension 2 are not squaregraphs~\cite{BaChEp10}. The last one has dimension $d=3$ as it contains an hypercube $Q_3$.

We provide a list of properties satisfied by median graphs. In particular, we define the notion of $\Theta$-classes which is a key ingredient of several existing algorithms~\cite{BeChChVa20,HaImKl99,ImKlMu99}.

In general graphs, all gated subgraphs are convex. The reverse is true in median graphs.
\begin{lemma}[Convex$\Leftrightarrow$Gated~\cite{BaCh08,BeChChVa20}]
Any convex subgraph of a median graph is gated.
\end{lemma}

\begin{figure}[t]

\begin{subfigure}[b]{0.33\columnwidth}
\centering
\scalebox{0.6}{\begin{tikzpicture}


\node[draw, circle, minimum height=0.2cm, minimum width=0.2cm, fill=black] (P1) at (3,4) {};
\node[draw, circle, minimum height=0.2cm, minimum width=0.2cm, fill=black] (P2) at (2,2.5) {};
\node[draw, circle, minimum height=0.2cm, minimum width=0.2cm, fill=black] (P3) at (4,2.5) {};
\node[draw, circle, minimum height=0.2cm, minimum width=0.2cm, fill=black] (P4) at (1.2,1) {};
\node[draw, circle, minimum height=0.2cm, minimum width=0.2cm, fill=black] (P5) at (2.8,1) {};


\draw[line width = 1.4pt] (P1) -- (P2);
\draw[line width = 1.4pt] (P1) -- (P3);
\draw[line width = 1.4pt] (P2) -- (P4);
\draw[line width = 1.4pt] (P2) -- (P5);

%
%

\end{tikzpicture}}
\caption{A tree, $d=1$}
\end{subfigure}
\begin{subfigure}[b]{0.33\columnwidth}
\centering
\scalebox{0.6}{\begin{tikzpicture}


\node[draw, circle, minimum height=0.2cm, minimum width=0.2cm, fill=black] (P11) at (1,1) {};
\node[draw, circle, minimum height=0.2cm, minimum width=0.2cm, fill=black] (P12) at (1,2.5) {};

\node[draw, circle, minimum height=0.2cm, minimum width=0.2cm, fill=black] (P21) at (3,1) {};
\node[draw, circle, minimum height=0.2cm, minimum width=0.2cm, fill=black] (P22) at (3,2.5) {};
\node[draw, circle, minimum height=0.2cm, minimum width=0.2cm, fill=black] (P23) at (3,4) {};

\node[draw, circle, minimum height=0.2cm, minimum width=0.2cm, fill=black] (P31) at (5,1) {};
\node[draw, circle, minimum height=0.2cm, minimum width=0.2cm, fill=black] (P32) at (5,2.5) {};
\node[draw, circle, minimum height=0.2cm, minimum width=0.2cm, fill=black] (P33) at (5,4) {};

\node[draw, circle, minimum height=0.2cm, minimum width=0.2cm, fill=black] (P41) at (1.3,3.8) {};
\node[draw, circle, minimum height=0.2cm, minimum width=0.2cm, fill=black] (P42) at (1.3,5.3) {};
\node[draw, circle, minimum height=0.2cm, minimum width=0.2cm, fill=black] (P43) at (-0.7,3.8) {};


\draw[line width = 1.4pt] (P11) -- (P12);
\draw[line width = 1.4pt] (P11) -- (P21);
\draw[line width = 1.4pt] (P12) -- (P22);
\draw[line width = 1.4pt] (P21) -- (P22);

\draw[line width = 1.4pt] (P21) -- (P31);
\draw[line width = 1.4pt] (P22) -- (P32);
\draw[line width = 1.4pt] (P31) -- (P32);

\draw[line width = 1.4pt] (P22) -- (P23);
\draw[line width = 1.4pt] (P23) -- (P33);
\draw[line width = 1.4pt] (P32) -- (P33);

\draw[line width = 1.4pt] (P22) -- (P41);
\draw[line width = 1.4pt] (P12) -- (P43);
\draw[line width = 1.4pt] (P23) -- (P42);
\draw[line width = 1.4pt] (P41) -- (P43);
\draw[line width = 1.4pt] (P41) -- (P42);

\end{tikzpicture}}
\caption{A cogwheel, $d=2$}
\end{subfigure}
\begin{subfigure}[b]{0.33\columnwidth}
\centering
\scalebox{0.6}{\begin{tikzpicture}


\node[draw, circle, minimum height=0.2cm, minimum width=0.2cm, fill=black] (P31) at (5,1) {};
\node[draw, circle, minimum height=0.2cm, minimum width=0.2cm, fill=black] (P32) at (5,2.5) {};

\node[draw, circle, minimum height=0.2cm, minimum width=0.2cm, fill=black] (P41) at (7,1) {};
\node[draw, circle, minimum height=0.2cm, minimum width=0.2cm, fill=black] (P42) at (7,2.5) {};

\node[draw, circle, minimum height=0.2cm, minimum width=0.2cm, fill=black] (P51) at (9,1) {};
\node[draw, circle, minimum height=0.2cm, minimum width=0.2cm, fill=black] (P52) at (9,2.5) {};

\node[draw, circle, minimum height=0.2cm, minimum width=0.2cm, fill=black] (P61) at (8.0,1.4) {};
\node[draw, circle, minimum height=0.2cm, minimum width=0.2cm, fill=black] (P62) at (8.0,2.9) {};
\node[draw, circle, minimum height=0.2cm, minimum width=0.2cm, fill=black] (P63) at (10.0,1.4) {};
\node[draw, circle, minimum height=0.2cm, minimum width=0.2cm, fill=black] (P64) at (10.0,2.9) {};

\node[draw, circle, minimum height=0.2cm, minimum width=0.2cm, fill=black] (P72) at (8.0,4.4) {};
\node[draw, circle, minimum height=0.2cm, minimum width=0.2cm, fill=black] (P74) at (10.0,4.4) {};


\draw[line width = 1.4pt] (P31) -- (P32);
\draw[line width = 1.4pt] (P31) -- (P41);
\draw[line width = 1.4pt] (P32) -- (P42);
\draw[line width = 1.4pt] (P41) -- (P42);

\draw[line width = 1.4pt](P41) -- (P51);
\draw[line width = 1.4pt] (P42) -- (P52);
\draw[line width = 1.4pt] (P51) -- (P52);

\draw[line width = 1.4pt] (P41) -- (P61);
\draw[line width = 1.4pt] (P42) -- (P62);
\draw[line width = 1.4pt] (P51) -- (P63);
\draw[line width = 1.4pt] (P52) -- (P64);
\draw[line width = 1.4pt] (P61) -- (P62);
\draw[line width = 1.4pt] (P61) -- (P63);
\draw[line width = 1.4pt] (P62) -- (P64);
\draw[line width = 1.4pt] (P63) -- (P64);

\draw[line width = 1.4pt] (P62) -- (P72);
\draw[line width = 1.4pt] (P64) -- (P74);

\end{tikzpicture}}
\caption{$d=3$}
\end{subfigure}

\caption{Examples of median graphs}
\label{fig:median_examples}
\end{figure}
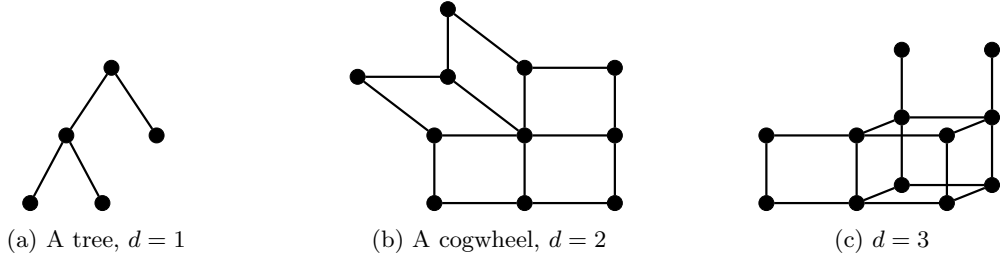

We remind the notion of $\Theta$-class, which is well explained in~\cite{BeChChVa20}, and enumerate some properties related to it. We say that the edges $uv$ and $xy$ are in relation $\Theta_0$ if they form a square $uvyx$, where $uv$ and $xy$ are opposite edges. Then, $\Theta$ refers to the reflexive and transitive closure of relation $\Theta_0$. The classes of the equivalence relation $\Theta$ are denoted by $E_1,\ldots,E_q$.  We denote by $\mathcal{E}$ the set of $\Theta$-classes: $\mathcal{E} = \set{E_1,\ldots,E_q}$. Parameter $q$ is less than the number of vertices $n$. To avoid confusions, let us highlight that parameter $q$ is different from the dimension $d$: for example, on trees, $d=1$ whereas $q = n-1$. Moreover, the dimension $d$ is at most $\lfloor \log n \rfloor$.

\begin{lemma}[$\Theta$-classes in linear time~\cite{BeChChVa20}]
There is an algorithm which computes the $\Theta$-classes $E_1,\ldots,E_q$ of a median graph in linear time $O(\card{E})$.
\end{lemma}

In median graphs, each class $E_i$, $1\le i\le q$, is a perfect matching cutset and its two sides $H_i'$ and $H_i''$ verify nice properties, that are presented below.

\begin{lemma}[Halfspaces of $E_i$~\cite{BeChChVa20,HaImKl99,Mu80}]
For any $1\le i\le q$, the graph $G$ deprived of edges of $E_i$, {\em i.e.} $G\backslash E_i = (V,E\backslash E_i)$, has two connected components $H_i'$ and $H_i''$, called \textit{halfspaces}. Edges of $E_i$ form a matching: they have no endpoint in common. Halfspaces satisfy the following properties:
\begin{itemize}
\item Both $H_i'$ and $H_i''$ are convex/gated.
\item If $uv$ is an edge of $E_i$ with $u \in H_i'$ and $v \in H_i''$, then $H_i' = W(u,v) = \set{x \in V: d(x,u) < d(x,v)}$ and $H_i'' = W(v,u) = \set{x \in V: d(x,v) < d(x,u)}$.
\end{itemize}
\label{le:halfspaces}
\end{lemma}

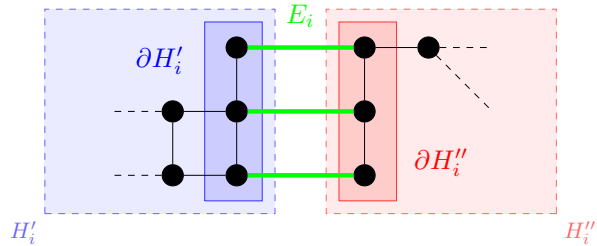
\begin{figure}[h]
\centering
\scalebox{0.85}{\begin{tikzpicture}


\draw [dashed, color = white!40!blue, fill = white!92!blue] (1.0,0.4) -- (4.6,0.4) -- (4.6,3.6) -- (1.0,3.6) -- (1.0,0.4) node[below left] {$H_i'$};
\draw [dashed, color = white!40!red, fill = white!92!red] (9.0,0.4) -- (9.0,3.6) -- (5.4,3.6) -- (5.4,0.4) -- (9.0,0.4) node[below right] {$H_i''$};
\draw [color = white!20!blue, fill = white!80!blue] (3.5,0.6) -- (4.4,0.6) -- (4.4,3.4) -- (3.5,3.4) -- (3.5,0.6);
\draw [color = white!20!red, fill = white!80!red] (5.6,0.6) -- (6.5,0.6) -- (6.5,3.4) -- (5.6,3.4) -- (5.6,0.6);


\node[draw, circle, minimum height=0.2cm, minimum width=0.2cm, fill=black] (P11) at (4,1) {};
\node[draw, circle, minimum height=0.2cm, minimum width=0.2cm, fill=black] (P12) at (4,2) {};
\node[draw, circle, minimum height=0.2cm, minimum width=0.2cm, fill=black] (P13) at (4,3) {};

\node[draw, circle, minimum height=0.2cm, minimum width=0.2cm, fill=black] (P21) at (6,1) {};
\node[draw, circle, minimum height=0.2cm, minimum width=0.2cm, fill=black] (P22) at (6,2) {};
\node[draw, circle, minimum height=0.2cm, minimum width=0.2cm, fill=black] (P23) at (6,3) {};

\node[draw, circle, minimum height=0.2cm, minimum width=0.2cm, fill=black] (P01) at (3,1) {};
\node[draw, circle, minimum height=0.2cm, minimum width=0.2cm, fill=black] (P02) at (3,2) {};

\node[draw, circle, minimum height=0.2cm, minimum width=0.2cm, fill=black] (P33) at (7,3) {};


\draw[line width = 2pt, color = green] (P11) -- (P21);
\draw[line width = 2pt, color = green] (P12) -- (P22);
\draw[line width = 2pt, color = green] (P13) -- (P23);

\draw (P11) -- (P12);
\draw (P12) -- (P13);
\draw (P21) -- (P22);
\draw (P22) -- (P23);

\draw (P01) -- (P11);
\draw (P02) -- (P12);
\draw (P01) -- (P02);
\draw[dashed] (P02) -- (2,2);
\draw[dashed] (P01) -- (2,1);

\draw (P23) -- (P33);
\draw[dashed] (P33) -- (8,3);
\draw[dashed] (P33) -- (8,2);


\node[scale=1.2, color = blue] at (2.8,2.8) {$\partial H_i'$};
\node[scale=1.2, color = red] at (7.2,1.2) {$\partial H_i''$};
\node[scale=1.2, color = green] at (5.0,3.5) {$E_i$};

\end{tikzpicture}}
\caption{A class $E_i$ with sets $H_i', H_i'', \partial H_i', \partial H_i''$}
\label{fig:halfspaces}
\end{figure}

We denote by $\partial H_i'$ the subset of $H_i'$ containing the vertices which are adjacent to a vertex in $H_i''$. Put differently, set $\partial H_i'$ is made up of vertices of $H_i'$ which are endpoints of edges in $E_i$. Similarly, set $\partial H_i''$ contains the vertices of $H_i''$ which are adjacent to $H_i'$. We say these sets are the \textit{boundaries} of halfspaces $H_i'$ and $H_i''$ respectively.

\begin{lemma}[Boundaries~\cite{BeChChVa20,HaImKl99,Mu80}]
Both $\partial H_i'$ and $\partial H_i''$ are convex/gated. Moreover, the edges of $E_i$ define an isomorphism between $\partial H_i'$ and $\partial H_i''$.
\label{le:boundaries}
\end{lemma}

As a consequence, suppose $uv$ and $u'v'$ belong to $E_i$: if $uu'$ is an edge and belongs to class $E_j$, then $vv'$ is an edge of $E_j$ too. We terminate this list of lemmas with a last property dealing with the orientation of edges  from a canonical basepoint $v_0 \in V$. The \textit{$v_0$-orientation} of the edges of $G$ according to $v_0$ is such that, for any edge $uv$, the orientation is $\vv{uv}$ if $d(v_0,u) < d(v_0,v)$. Indeed, we cannot have $d(v_0,u) = d(v_0,v)$ as $G$ is bipartite.

\begin{lemma}[Orientation~\cite{BeChChVa20}]
All edges can be oriented according to any canonical basepoint $v_0$.
\end{lemma}

From now on, we suppose that an arbitrary basepoint $v_0 \in V$ has been selected and we refer automatically to the $v_0$-orientation when we mention incoming or outgoing edges.

\subsection{Orthogonal $\Theta$-classes and hypercubes}

We present now an important notion on median graphs. In~\cite{Ko09}, Kovse studied a relationship between \textit{splits} which refer to the halfspaces of $\Theta$-classes. It says that two splits $\set{H_i',H_i''}$ and $\set{H_j',H_j''}$ are \textit{incompatible} if the four sets $H_i' \cap H_j'$, $H_i'' \cap H_j'$, $H_i' \cap H_j''$, and $H_i'' \cap H_j''$ are nonempty. We use different notation to characterize this relationship which makes more sense in our context.

\begin{definition}[Orthogonal classes]
We say that classes $E_i$ and $E_j$ are {\em orthogonal} ($E_i \perp E_j$) if there is a square $uvyx$ in $G$, where $uv,xy \in E_i$ and $ux,vy \in E_j$.
\end{definition}

We observe that classes $E_i$ and $E_j$ are orthogonal if and only if the splits produced by their halfspaces are incompatible.

\begin{lemma}[Orthogonal$\Leftrightarrow$Incompatible]
$E_i$ and $E_j$ are orthogonal if and only if
$\set{H_i',H_i''}$ and $\set{H_j',H_j''}$ are incompatible.
\label{le:perp_incomp}
\end{lemma}
\begin{proof}
First, if $E_i$ and $E_j$ are orthogonal, there is a square containing the edges of these two classes. The four vertices belong respectively to the four sets $H_i' \cap H_j'$, $H_i'' \cap H_j'$, $H_i' \cap H_j''$, and $H_i'' \cap H_j''$. Consequently, these sets are nonempty and the splits are incompatible.

Second, suppose the four sets $H_i' \cap H_j'$, $H_i'' \cap H_j'$, $H_i' \cap H_j''$, and $H_i'' \cap H_j''$ are nonempty. Consider only the vertices in $H_i'$. As $H_i' \cap H_j'$ and $H_i' \cap H_j''$ are nonempty, we take arbitrarily one vertex $x_1$ in $H_i' \cap H_j'$ and one vertex $x_2$ in $H_i' \cap H_j''$. By convexity of $H_i'$ (Lemma~\ref{le:halfspaces}), any shortest $(x_1,x_2)$-path passes through an edge of $E_j$ with its two endpoints in $H_i'$: this edge is denoted by $y_1y_2$, where $y_1 \in H_i' \cap \partial H_j'$ and $y_2 \in H_i' \cap \partial H_j''$. If we consider now only the vertices in $H_i''$, we can point out an edge $y_3y_4$ with the same method, where $y_3 \in H_i'' \cap \partial H_j'$ and $y_4 \in H_i'' \cap \partial H_j''$. By convexity of $\partial H_j'$ (Lemma~\ref{le:boundaries}), any shortest $(y_1,y_3)$-path passes through an edge of $E_i$ with two endpoints in $\partial H_j'$: we denote this edge by $z_1z_3$, where $z_1 \in \partial H_i' \cap \partial H_j'$ and $z_3 \in \partial H_i'' \cap \partial H_j'$. Let $z_2$ (resp. $z_4$) be the neighbor of $z_1$ (resp. $z_3$) in $\partial H_j''$. We have $z_1z_2, z_3z_4 \in E_j$ and $z_1z_3 \in E_i$. As $E_j$ defines an isomorphism (Lemma~\ref{le:boundaries}) between $\partial H_j'$ and $\partial H_j''$, then $z_2z_4 \in E$. We obtain a square $z_1z_2z_4z_3$ with edges belonging to classes $E_i$ and $E_j$.
\end{proof}

This results proves implicitely that the splits $\set{H_i',H_i''}$ and $\set{H_j',H_j''}$ are incompatible if and only if the following four sets formed with boundaries $\partial H_i' \cap \partial H_j'$, $\partial H_i'' \cap \partial H_j'$, $\partial H_i' \cap \partial H_j''$, and $\partial H_i'' \cap \partial H_j''$ are nonempty.

We pursue with a property on orthogonal classes: if two edges of two orthogonal classes $E_i$ and $E_j$ are incident, they belong to a common square. Even if the result was already proposed in~\cite{BaCo93}, we present a different proof consistent with the notions evoked earlier.

\begin{lemma}[Squares~\cite{BaCo93}]
Let $xu \in E_i$ and $uy \in E_j$. If $E_i$ and $E_j$ are orthogonal, then there is a vertex $v$ such that $uyvx$ is a square.
\label{le:squares}
\end{lemma}
\begin{proof}
We say w.l.o.g. that $u \in H_i' \cap H_j'$. Then, $x \in H_i'' \cap H_j'$ and $y \in H_i' \cap H_j''$. Let $z$ be an arbitrary vertex of $H_i'' \cap H_j''$ which is nonempty (Lemma~\ref{le:perp_incomp}). We fix $v = m(x,y,z)$. We prove that vertex $v$ is different from $x$, $u$, and $y$. 

As $H_i''$ is convex, $I(x,z) \subseteq H_i''$ and $v \in H_i''$. Consequently, $v \neq u,y$. Similarly, as $H_j''$ is convex, $I(y,z) \subseteq H_j''$, so $v\neq x$. The median $v = m(x,y,z)$ belongs to $I(x,y)$ and the distance $d(x,y)$ is at most 2. As $v \neq x,u,y$, vertex $v$ belongs to a shortest $(x,y)$-path of length 2 which does not contain $u$. As a consequence, $v$ is adjacent to $x$ and $y$.
\end{proof}

Now, we focus on set of classes which are pairwise orthogonal.

\begin{definition}[Pairwise Orthogonal Family]
We say that a set of classes $X \subseteq \mathcal{E}$ is a {\em Pairwise Orthogonal Family (POF)} if for any pair $E_j,E_h \in X$, we have $E_j \perp E_h$.
\end{definition}

The empty set is considered as a POF. The notion of POF is strongly related to the induced hypercubes in median graphs. First, observe that all $\Theta$-classes of a median graph form a POF if and only if the graph is an hypercube of dimension $\log n$~\cite{Ko09,MoMuRo98}. Second, if all classes of a POF $X$ are adjacent to $v \in V$, there is an hypercube containing $v$ and its edges belong to classes in $X$.

\begin{lemma}[POFs adjacent to a vertex]
Let $X$ be a POF, $v \in V$, and assume that for each class $E_i \in X$, there is an edge of $E_i$ adjacent to $v$. There exists an hypercube $Q$ containing vertex $v$ and all edges of $X$ adjacent to $v$. Moreover, the $\Theta$-classes of the edges of $Q$ are the classes of $X$.
\label{le:pof_adjacent}
\end{lemma}
\begin{proof}
We proceed inductively. If $\card{X} = 1$, $X = \set{E_i}$, then the edge $e_i$ adjacent to $v$ is an hypercube of dimension 1 and $v$ is one of its endpoints. If $\card{X} = 2$, $X = \set{E_i,E_j}$, Lemma~\ref{le:squares} shows that there is a square containing $v$ and the two edges of $X$ adjacent to $v$.

Assume that for a POF $X$ adjacent to $v$, $\card{X} \le k$, there is an hypercube $Q$ containing $v$ and such that its edges belong to the classes in $X$. Consider now a POF $X$, $\card{X} = k+1$, say w.l.o.g. $X = \set{E_1,\ldots,E_{k+1}}$, adjacent to $v$. We denote by $e_1,\ldots,e_{k+1}$ the edges belonging respectively to classes $E_1,\ldots,E_{k+1}$ which are adjacent to $v$. Indeed, we cannot have two edges of the same class $E_i$ adjacent to $v$ as any $E_i$ is a matching (Lemma~\ref{le:halfspaces}). Let $X' = \set{E_1,\ldots,E_k}$ and $Q'$ be the hypercube containing $v$ and admitting edges in $X'$. We say w.l.o.g. that vertex $v$ belongs to the boundary $\partial H_{k+1}'$ of class $E_{k+1}$.

We apply Lemma~\ref{le:squares} to all pairs of edges $(e_i,e_{k+1})$ with $1\le i\le k$. They have a common endpoint $v$. All these pairs belong in fact to a square of classes $\set{E_i,E_k}$. As a consequence, all vertices in $Q'$ which are at distance 1 from $v$ belong to $\partial H_{k+1}'$. We can pursue this reasoning for the vertices of $Q'$ at distance 2 from $v$: the edges of $Q'$ connecting vertices at distance 1 with vertices at distance 2 from $v$ belong to $X'$. Moreover, the vertices at distance 1 are endpoints of edges in $E_{k+1}$. According to Lemma~\ref{le:squares}, the vertices of $Q'$ at distance 2 from $v$ belong to $\partial H_{k+1}'$. Finally, after multiple applications of Lemma~\ref{le:squares}, we prove that all vertices in $Q'$ belong to $\partial H_{k+1}'$. 

As class $E_{k+1}$ forms an isomorphism between $\partial H_{k+1}'$ and $\partial H_{k+1}''$, the matching $E_{k+1}$ connects a hypercube $Q'$ with a isomorphic hypercube $Q''$ in $\partial H_{k+1}''$. The induced subgraph on $Q'$ and $Q''$ is the Cartesian product between $Q_k$ and $K_2$: it forms an hypercube $Q$ of dimension $k+1$. As $v \in Q'$, then $v \in Q$. Moreover, the $\Theta$-classes represented in $Q$ are exactly $X' \cup \set{E_{k+1}} = X$.
\end{proof}

In fact, there is a bijection between the POFs and the vertices of a median graph. Given a POF $X$, the vertex which is associated with it is the farthest-to-$v_0$ vertex of the closest-to-$v_0$ hypercube formed with $\Theta$-classes $X$.

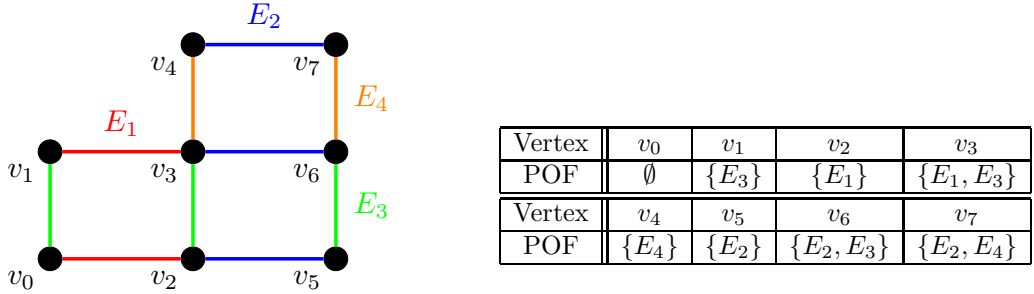
\begin{figure}[h]
\begin{subfigure}[b]{0.45\columnwidth}
\centering
\scalebox{0.95}{\begin{tikzpicture}


\node[draw, circle, minimum height=0.2cm, minimum width=0.2cm, fill=black] (P11) at (1,1) {};
\node[draw, circle, minimum height=0.2cm, minimum width=0.2cm, fill=black] (P12) at (1,2.5) {};

\node[draw, circle, minimum height=0.2cm, minimum width=0.2cm, fill=black] (P21) at (3,1) {};
\node[draw, circle, minimum height=0.2cm, minimum width=0.2cm, fill=black] (P22) at (3,2.5) {};
\node[draw, circle, minimum height=0.2cm, minimum width=0.2cm, fill=black] (P23) at (3,4) {};

\node[draw, circle, minimum height=0.2cm, minimum width=0.2cm, fill=black] (P31) at (5,1) {};
\node[draw, circle, minimum height=0.2cm, minimum width=0.2cm, fill=black] (P32) at (5,2.5) {};
\node[draw, circle, minimum height=0.2cm, minimum width=0.2cm, fill=black] (P33) at (5,4) {};


\draw[line width = 1.4pt, color = green] (P11) -- (P12);
\draw[line width = 1.4pt, color = red] (P11) -- (P21);
\draw[line width = 1.4pt, color = red] (P12) -- (P22);
\draw[line width = 1.4pt, color = green] (P21) -- (P22);

\draw[line width = 1.4pt, color = blue] (P21) -- (P31);
\draw[line width = 1.4pt, color = blue] (P22) -- (P32);
\draw[line width = 1.4pt, color = green] (P31) -- (P32);

\draw[line width = 1.4pt, color = orange] (P22) -- (P23);
\draw[line width = 1.4pt, color = blue] (P23) -- (P33);
\draw[line width = 1.4pt, color = orange] (P32) -- (P33);


\node[scale=1.2, color = red] at (2.0,2.9) {$E_1$};
\node[scale=1.2, color = blue] at (4.0,4.4) {$E_2$};
\node[scale=1.2, color = green] at (5.5,1.75) {$E_3$};
\node[scale=1.2, color = orange] at (5.5,3.25) {$E_4$};

\node[scale = 1.2] at (0.6,0.7) {$v_0$};
\node[scale = 1.2] at (0.6,2.2) {$v_1$};
\node[scale = 1.2] at (2.6,0.7) {$v_2$};
\node[scale = 1.2] at (2.6,2.2) {$v_3$};
\node[scale = 1.2] at (2.6,3.7) {$v_4$};
\node[scale = 1.2] at (4.6,0.7) {$v_5$};
\node[scale = 1.2] at (4.6,2.2) {$v_6$};
\node[scale = 1.2] at (4.6,3.7) {$v_7$};

\end{tikzpicture}}
\end{subfigure}
\begin{subfigure}[b]{0.51\columnwidth}
\centering
$\begin{array}{|c||c|c|c|c|}
\hline
\mbox{Vertex} & v_0 & v_1 & v_2 & v_3\\
\hline
 \mbox{POF} & \emptyset & \set{E_3} & \set{E_1} & \set{E_1,E_3}\\
\hline
\hline
\mbox{Vertex} & v_4 & v_5 & v_6 & v_7\\
\hline 
\mbox{POF} & \set{E_4} & \set{E_2} & \set{E_2,E_3} & \set{E_2,E_4}\\
\hline
\end{array}$
~

~

~
\end{subfigure}
\caption{Illustration of the bijection between $V$ and the set of POFs.}
\label{fig:vertices_pof}
\end{figure}

\begin{lemma}[POFs and hypercubes~\cite{BaChDrKo06,BaQuSaMa02,Ko09}]
Consider an arbitrary canonical basepoint $v_0 \in V$ and the $v_0$-orientation for the median graph $G$. Given a vertex $v \in V$, let $N^-(v)$ be the set of edges going into $v$ according to the $v_0$-orientation. Let $\mathcal{E}^-(v)$ be the classes of the edges in $N^-(v)$. The following propositions are true:
\begin{itemize}
\item For any vertex $v\in V$, $\mathcal{E}^-(v)$ is a POF. Moreover, vertex $v$ and the edges of $N^-(v)$ belong to an induced hypercube formed by the classes $\mathcal{E}^-(v)$.
\item For any POF $X$, there is a unique vertex $v_X$ such that $\mathcal{E}^-(v_X) = X$. Vertex $v_X$ is the closest-to-$v_0$ vertex $v$ such that $X \subseteq \mathcal{E}^-(v)$.
\item The number of POFs in $G$ is equal to the number $n$ of vertices.
\end{itemize}
\label{le:pof_hypercube}
\end{lemma}

This result highlights a bijection between the vertices $V$ and the POFs of $G$. An example is given in Figure~\ref{fig:vertices_pof} with a graph of dimension $d=2$. Edges are colored in function of their $\Theta$-class. Vertex $v_0$ is the canonical basepoint. For example, $v_1v_3 \in E_1$. We associate with any POF $X$ of $G$ the vertex $v$ satisfying $\mathcal{E}^-(v) = X$ with the $v_0$-orientation. Obviously, the empty POF is associated with $v_0$ which has no incoming edges.

This bijection can be used to enumerate the POFs of a median graph in linear time~\cite{BaQuSaMa02,Ko09}. Given a basepoint $v_0$, we say that the \textit{basis} (resp. \textit{anti-basis}) of an induced hypercube $Q$ is the single vertex $v$ such that all edges of the hypercube adjacent to $v$ are outgoing (resp. ingoing) from $v$. Said differently, the basis of $Q$ is its closest-to-$v_0$ vertex and its anti-basis is its farthest-to-$v_0$ vertex. What Lemma~\ref{le:pof_hypercube} states is also that we can associate with any POF $X$ an hypercube $Q_X$ which contains exactly the classes $X$ and admits $v_X$ as its anti-basis. Moreover, the hypercube $Q_X$ is the closest-to-$v_0$ hypercube formed with the classes in $X$. Figure~\ref{fig:ingoing_edges} shows a vertex $v$ with its ingoing and outgoing edges with the $v_0$-orientation. The dashed edges represent the hypercube with anti-basis $v$ and POF $\mathcal{E}^-(v)$.

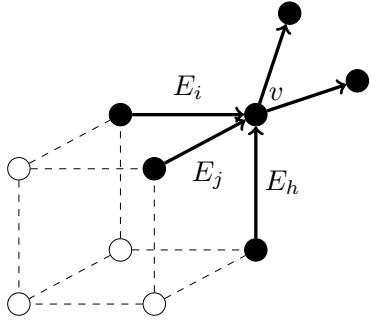
\begin{figure}[h]
\centering
\scalebox{0.9}{\begin{tikzpicture}


\node[draw, circle, minimum height=0.2cm, minimum width=0.2cm, fill=black] (P0) at (6,6) {};

\node[draw, circle, minimum height=0.2cm, minimum width=0.2cm, fill=black] (P1) at (4,6) {};
\node[draw, circle, minimum height=0.2cm, minimum width=0.2cm, fill=black] (P2) at (4.5,5.2) {};
\node[draw, circle, minimum height=0.2cm, minimum width=0.2cm, fill=black] (P3) at (6,4) {};

\node[draw, circle, minimum height=0.2cm, minimum width=0.2cm, fill=black] (P1') at (6.5,7.5) {};
\node[draw, circle, minimum height=0.2cm, minimum width=0.2cm, fill=black] (P2') at (7.5,6.5) {};

\node[draw, circle, minimum height=0.2cm, minimum width=0.2cm] (P4) at (2.5,5.2) {};
\node[draw, circle, minimum height=0.2cm, minimum width=0.2cm] (P5) at (4.5,3.2) {};
\node[draw, circle, minimum height=0.2cm, minimum width=0.2cm] (P6) at (4,4) {};
\node[draw, circle, minimum height=0.2cm, minimum width=0.2cm] (P7) at (2.5,3.2) {};


\draw[->, line width = 1.4pt] (P1) -- (P0);
\draw[->, line width = 1.4pt] (P2) -- (P0);
\draw[->, line width = 1.4pt] (P3) -- (P0);

\draw[->, line width = 1.4pt] (P0) -- (P1');
\draw[->, line width = 1.4pt] (P0) -- (P2');

\draw[dashed] (P4) -- (P1);
\draw[dashed] (P4) -- (P2);
\draw[dashed] (P5) -- (P2);
\draw[dashed] (P5) -- (P3);
\draw[dashed] (P6) -- (P1);
\draw[dashed] (P6) -- (P3);
\draw[dashed] (P4) -- (P7);
\draw[dashed] (P5) -- (P7);
\draw[dashed] (P6) -- (P7);


\node[scale=1.2] at (5.0,6.4) {$E_i$};
\node[scale=1.2] at (5.3,5.1) {$E_j$};
\node[scale=1.2] at (6.4,5.0) {$E_h$};

\node[scale = 1.2] at (6.3,6.3) {$v$};

\end{tikzpicture}}
\caption{Edges ingoing to a vertex $v$ belong to an hypercube made up of their classes $\mathcal{E}^-(v) = \set{E_i,E_j,E_h}$ and with anti-basis $v$.}
\label{fig:ingoing_edges}
\end{figure} 

This observation implies that the number of POFs is less than the number of hypercubes in $G$. We remind a formula establishing a relationship between the number of POFs and the number of hypercubes in the literature. Let $\alpha(G)$ (resp. $\beta(G)$) be the number of hypercubes (resp. POFs) in $G$. Let $\beta_i(G)$ be the number of POFs of cardinality $i \le d$ in $G$. According to~\cite{BaQuSaMa02,Ko09}, we have:
\begin{equation}
\alpha(G) = \sum_{i=0}^d 2^i\beta_i(G)
\label{eq:number_hypercubes}
\end{equation}

Equation~\eqref{eq:number_hypercubes} produces a natural upper bound for the number of hypercubes.

\begin{lemma}[Number of hypercubes]
$\alpha(G)\le 2^dn$.
\label{le:number_hypercubes}
\end{lemma}
\begin{proof}
This is the consequence of $\beta(G) = \sum_{i=0}^d \beta_i(G) = n$ (Lemma~\ref{le:pof_hypercube}).
\end{proof}

Each hypercube in the median graph $G$ can be defined with only its anti-basis $v$ and the edges $\widehat{N}$ of the hypercube that are adjacent and going into $v$ according the $v_0$-orientation. These edges are a subset of $N^-(v)$: $\widehat{N} \subseteq N^-(v)$. Conversely, given a vertex $v$, each subset of $N^-(v)$ produces an hypercube which admits $v$ as an anti-basis (this hypercube is a sub-hypercube of the one obtained with $v$ and $N^-(v)$, Lemma~\ref{le:pof_hypercube}). Another possible bijection is to consider an hypercube as a pair composed of its anti-basis $v$ and the $\Theta$-classes $\widehat{\mathcal{E}}$ of the edges in $\widehat{N}$.

From the previous results and Lemma~\ref{le:number_hypercubes}, we deduce an algorithm enumerating the hypercubes in $G$ in time $O(d2^dn)$.

\begin{lemma}[Enumeration of hypercubes]
We can enumerate all triplets $(v,u,\widehat{\mathcal{E}})$, where $v$ is the anti-basis of an hypercube $Q$, $u$ its basis, and $\widehat{\mathcal{E}}$ the $\Theta$-classes of the edges of $Q$ in time $O(d2^dn)$. Moreover, the list obtained fulfils the following partial order: if $d(v_0,v) < d(v_0,v')$, then any triplet $(v,u,\widehat{\mathcal{E}})$ containing $v$ appears before any triplet $(v',u',\widehat{\mathcal{E}}')$ containing $v'$.
\label{le:enum_hypercubes}
\end{lemma}
\begin{proof}
We run a BFS from a canonical $v_0$ to obtain the $v_0$-orientation. For any vertex $v$ visited, we list all subsets of $N^-(v)$. Each of these subsets $\widehat{N}$ correspond to an hypercube with anti-basis $v$. For each $\widehat{N}$, we determine the $\Theta$-classes $\widehat{\mathcal{E}}$ of its edges. For any pair $(v,\widehat{\mathcal{E}})$, we compute the basis of this hypercube: we start a walk from $v$ and we traverse once an edge for each class in $\widehat{\mathcal{E}}$ in any order. The vertex we visit at the end of the walk is the basis $u$. At the end of the execution, we have a list of triplets $(v,u,\widehat{\mathcal{E}})$, where $\widehat{\mathcal{E}} \subseteq \mathcal{E}^-(v)$. The execution time is equal the number of hypercubes multiplied by the size of the triplets representing them, which is upper-bounded by $d+1$. The partial order evoked in the statement of this lemma is due to the BFS which visits $v$ before $v'$ if $d(v_0,v) < d(v_0,v')$.
\end{proof}

The enumeration of hypercubes is executed in linear time for median graphs with $d=O(1)$. It will be used in the remainder to describe our algorithm computing the diameter.

\section{Failure of BFS-based heuristics on median graphs} \label{sec:sweep}

We prove in this section that \tsw\ and \fsw , two well-known linear time heuristics for the diameter in general graphs, do not determine the exact value of diameter in median graphs. We begin with a short introduction of these algorithms and then present two median graphs on which they are unsuccessful.

The heuristic \tsw\ consists in two successive BFS returning a distance $d(a_1,b_1)$ between two vertices $a_1$ and $b_1$ which is supposed to estimate the diameter of the graph. First, it starts from a random vertex of $G$ denoted by $r_1$. It computes a first BFS starting from $r_1$ to determine the farthest-to-$r_1$ vertex $a_1$. Formally, vertex $a_1$ verifies $d(r_1,a_1) = \max_{v \in V} d(r_1,v)$. Then, it computes a second BFS, starting from $a_1$ to determine the farthest-to-$a_1$ vertex $b_1$. Vertex $b_1$ verifies $d(a_1,b_1) = \max_{v \in V} d(a_1,v)$. Value $d(a_1,b_1)$ is returned.

The heuristic \fsw\ consists in four successive BFS. We start with a first \tsw\ which enables us to obtain vertices $a_1$ and $b_1$. Then, we determine vertex $r_2$ which is the middle of a shortest $(a_1,b_1)$-path. We compute a second \tsw\ starting from vertex $r_2$ which gives vertices $a_2$ and $b_2$. Value $d(a_2,b_2)$ is returned.

\begin{figure}[b]
\begin{subfigure}[b]{0.49\columnwidth}
\centering
\scalebox{0.9}{\begin{tikzpicture}


\node[draw, circle, minimum height=0.2cm, minimum width=0.2cm, fill=black] (P11) at (1,1) {};
\node[draw, circle, minimum height=0.2cm, minimum width=0.2cm, fill=black] (P12) at (1,2.5) {};

\node[draw, circle, minimum height=0.2cm, minimum width=0.2cm, fill=black] (P21) at (3,1) {};
\node[draw, circle, minimum height=0.2cm, minimum width=0.2cm, fill=black] (P22) at (3,2.5) {};
\node[draw, circle, minimum height=0.2cm, minimum width=0.2cm, fill=black] (P23) at (4.5,3.5) {};


\draw[line width = 1.4pt] (P11) -- (P12);
\draw[line width = 1.4pt] (P11) -- (P21);
\draw[line width = 1.4pt] (P12) -- (P22);
\draw[line width = 1.4pt] (P21) -- (P22);

\draw[line width = 1.4pt] (P22) -- (P23);


\node[scale=1.1] at (3.4,1.3) {$r_1$};
\node[scale=1.1] at (0.6,2.8) {$a_1$};
\node[scale=1.1] at (4.7,3.1) {$b_1$};
\node[scale=1.1] at (0.6,1.3) {$w$};

\end{tikzpicture}}
\caption{Graph $G^*$}
\label{subfig:2sweep}
\end{subfigure}
\begin{subfigure}[b]{0.49\columnwidth}
\centering
\scalebox{0.7}{\begin{tikzpicture}


\node[draw, circle, minimum height=0.2cm, minimum width=0.2cm, fill=black] (P11) at (1,1) {};
\node[draw, circle, minimum height=0.2cm, minimum width=0.2cm, fill=black] (P12) at (1,2.5) {};
\node[draw, circle, minimum height=0.2cm, minimum width=0.2cm, fill=black] (P13) at (1,4) {};

\node[draw, circle, minimum height=0.2cm, minimum width=0.2cm, fill=black] (P21) at (3,1) {};
\node[draw, circle, minimum height=0.2cm, minimum width=0.2cm, fill=black] (P22) at (3,2.5) {};
\node[draw, circle, minimum height=0.2cm, minimum width=0.2cm, fill=black] (P23) at (3,4) {};

\node[draw, circle, minimum height=0.2cm, minimum width=0.2cm, fill=black] (P31) at (5,1) {};
\node[draw, circle, minimum height=0.2cm, minimum width=0.2cm, fill=black] (P32) at (5,2.5) {};
\node[draw, circle, minimum height=0.2cm, minimum width=0.2cm, fill=black] (P33) at (5,4) {};

\node[draw, circle, minimum height=0.2cm, minimum width=0.2cm, fill=black] (P4) at (0,4.75) {};
\node[draw, circle, minimum height=0.2cm, minimum width=0.2cm, fill=black] (P5) at (0,0.25) {};
\node[draw, circle, minimum height=0.2cm, minimum width=0.2cm, fill=black] (P6) at (6,4.75) {};


\draw[line width = 1.4pt] (P11) -- (P12);
\draw[line width = 1.4pt] (P12) -- (P13);
\draw[line width = 1.4pt] (P11) -- (P21);
\draw[line width = 1.4pt] (P12) -- (P22);
\draw[line width = 1.4pt] (P13) -- (P23);
\draw[line width = 1.4pt] (P21) -- (P22);
\draw[line width = 1.4pt] (P22) -- (P23);
\draw[line width = 1.4pt] (P21) -- (P31);
\draw[line width = 1.4pt] (P22) -- (P32);
\draw[line width = 1.4pt] (P23) -- (P33);
\draw[line width = 1.4pt] (P31) -- (P32);
\draw[line width = 1.4pt] (P32) -- (P33);

\draw[line width = 1.4pt] (P11) -- (P5);
\draw[line width = 1.4pt] (P13) -- (P4);
\draw[line width = 1.4pt] (P33) -- (P6);


\node[scale=1.0] at (3.6,2.8) {$r_1,r_2$};
\node[scale=1.0] at (-0.6,4.5) {$a_1,a_2$};
\node[scale=1.0] at (5.6,0.7) {$b_1,b_2$};

\end{tikzpicture}}
\caption{Graph $H^*$}
\label{subfig:4sweep}
\end{subfigure}
\caption{Different steps of \tsw\ and \fsw\ on graphs $G^*$ and $H^*$.}
\label{fig:sweep}
\end{figure}
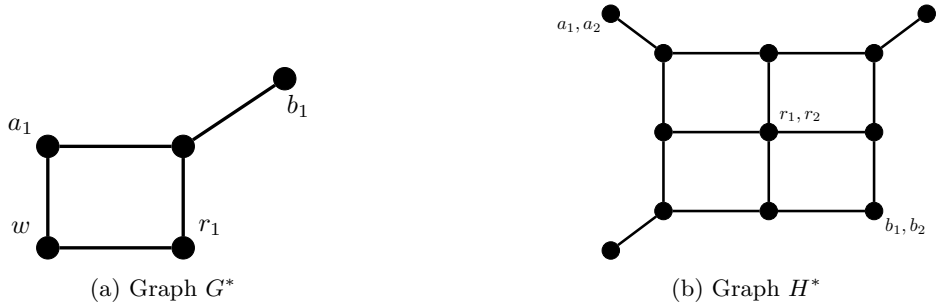

Unfortunately, both of these fast and simple algorithms fail to determine the diameter on median graphs. 

\begin{theorem}
There are two median graphs $G^*$ and $H^*$ of dimension 2 such that {\em 2}-\textsc{sweep} does not return the diameter on $G^*$ and {\em 4}-\textsc{sweep} does not return it on $H^*$.
\label{th:sweep}
\end{theorem}
\begin{proof}
We begin with graph $G^*$ and the execution of \tsw . Graph $G^*$ is a square with a pendant vertex. Figure~\ref{subfig:2sweep} represents $G^*$ and indicates the identity of $r_1$, $a_1$, and $b_1$. We start from the arbitrary vertex $r_1$. Its eccentricity is $\ecc(r_1) = 2$. Vertex $a_1$ verifies $d(r_1,a_1) = 2$ and is selected as the start of the next BFS. Similarly, $\ecc(a_1) = 2$ and $b_1$ is selected as the farthest-to-$a_1$ vertex. The distance returned is $d(a_1,b_1) = 2$ whereas $\diam(G^*) = d(b_1,w) = 3$.

We focus now on graph $H^*$ and a possible execution of \fsw . Graph $H^*$ is a $(3\times 3)$-grid with pendant vertices. Figure~\ref{subfig:4sweep} represents $H^*$ and reveals the identity of $r_1$, $a_1$, $b_1$, $a_2$ and $b_2$. We start at vertex $r_1$. Vertex $a_1$ maximizes $d(r_1,v)$ and $b_1$ maximizes $d(a_1,v)$. Then, vertex $r_2$ is a middle of a shortest $(a_1,b_1)$-path, but $r_2=r_1$. Thus, pair $a_1$ and $b_1$ can be selected as a potential diameter again and again. We have $a_1 = a_2$ and $b_1 = b_2$. In this case, the distance returned is $d(a_2,b_2) = 5$ while $\diam(G) = 6$.
\end{proof}

This result convinces us that a more involved algorithm is needed is we aim at determining not only the diameter of constant-dimension median graphs but also all its eccentricities.

\section{Computing the diameter in linear time for dimension $d=O(1)$} \label{sec:diameter}

We proceed in two steps. First, we compute labels for each hypercube of $G$: they characterize the shortest paths of $G$ starting at the basis of the hypercube and passing through its anti-basis (Section~\ref{subsec:labels}). Second, thanks to these labels, we reduce the diameter problem to \textsc{maximum-weighted disjoint sets} (Section~\ref{subsec:mwds}). We design an algorithm with running time $O^*(2^{d(\log(d) + 1) }n)$ which identifies the diameter and a diametral pair for median graphs. It will be extended to obtain all eccentricities with the same execution time in Section~\ref{sec:eccentricities}.

\subsection{Labels on hypercubes} \label{subsec:labels}

A naive approach consists, for each vertex $v \in V$, in computing a BFS to determine its distance to all other vertices of $G$. In this way, we associate with any vertex $v$ a label $\varphi(v)$ of size $O(n)$ indicating all these distances. Such a method produces necessarily a quadratic running time.

To determine the diameter of constant-dimension median graphs in linear time, we propose different labels which contain less information but enough for our problem. We label all pairs which contain each a vertex and a POF outgoing from this vertex. We remind that these pairs are in bijection with the set of hypercubes, according to Lemma~\ref{le:pof_hypercube}: the vertex is the basis of the hypercube and the POFs contains exactly the $\Theta$-classes of the hypercube. The idea is that the POFs indicate some ``direction'' in which we can go from the vertex. The label provides us with the distance of the longest shortest path we can find starting from the input vertex (the basis of the hypercube) and passing through the anti-basis of the hypercube. We will give more intuition on these labels when certain notions will be presented formally in this section.

As announced above, we fix an arbitrary canonical basepoint $v_0$ and for each class $E_i$, we say that the halfspace containing $v_0$ is $H_i'$.

\textbf{Signature.} Given two vertices $u,v \in V$, we define the set which contains the $\Theta$-classes separating $u$ from $v$.

\begin{definition}[Signature $\sigma_{u,v}$]
We say that the {\em signature} of the pair of vertices $u,v$, denoted by $\sigma_{u,v}$, is the set of classes $E_i$ such that $u$ and $v$ are separated in $G\backslash E_i$. In other words, $u$ and $v$ are in different halfspaces of $E_i$.
\end{definition}

All shortest $(u,v)$-paths contain exactly one edge for each class in $\sigma_{u,v}$.

\begin{theorem}
For any shortest $(u,v)$-path $P$, the edges in $P$ belong to classes in $\sigma_{u,v}$ and, for any $E_i \in \sigma_{u,v}$, there is exactly one edge of $E_i$ in path $P$.
\label{th:signature}
\end{theorem}
\begin{proof}
Each $\Theta$-class in $\sigma_{u,v}$ separates $u$ from $v$. Consequently, a shortest $(u,v)$-path necessarily passes through an edge of each of these classes. Let $E_i \in \sigma_{u,v}$. A shortest $(u,v)$-path passes through $E_i$ only once, otherwise we would come back to the first halfspace containing $u$ (say $H_i'$ w.l.o.g.) which is a contradiction with the convexity of $H_i'$ (Lemma~\ref{le:halfspaces}). Finally, suppose there is a class $E_j \notin \sigma_{u,v}$ represented in a shortest $(u,v)$-path. Vertices $u$ and $v$ are inside the same halfspace of $E_j$ which is a contradiction with the convexity of halfspaces again: the shortest $(u,v)$-path cannot traverse an edge of $E_j$.
\end{proof}

The converse also holds: if a path contains at most one edge of each class, it is a shortest path.

\begin{theorem}
A path containing at most one edge of each $\Theta$-class is a shortest path between its departure and its arrival.
\label{th:at_most_one}
\end{theorem}
\begin{proof}
We proceed by induction on the length $k$ of the path. If $k=1$, the path consists in one edge only, so it is obviously a shortest path. Suppose we have a path $P$ of length $k+1$ and that all paths of length $k$ having at most one edge of each class are shortest paths. Let $u$ denote the departure, $v$ the arrival and $w$ the vertex before the arrival. We know by induction that the section $P^{(u,w)}$ of the path $P$ between $u$ and $w$ is a shortest $(u,w)$-path. Let $E_j$ denote the $\Theta$-class of edge $wv$. There is no edge of class $E_j$ in path $P^{(u,w)}$, so all vertices in $P^{(u,w)}$ are in the same halfspace of $E_j$, say $H_j'$ w.l.o.g.. Vertex $w$ is the gate of $v$ in $H_j'$ because $d(v,w)=1$: no other vertex can be the gate of $v$, otherwise $d(v,w) \ge 2$. As $w \in I(u,v)$, path $P$ is a shortest $(u,v)$-path.
\end{proof}

The concept of signature can be generalized: given a set of edges, its signature is the set of $\Theta$-classes represented in that set. For example, the signature of a path is the set of classes which have at least one edge in this path. In this way, the signature $\sigma_{u,v}$ is also the signature of any shortest $(u,v)$-path. The signature of a hypercube is the set of $\Theta$-classes represented in its edges.

\textbf{Ladder sets.} Our idea is now to use POFs to characterize, given two vertices $u,v \in V$, the ``direction'' in which the shortest $(u,v)$-paths are oriented. We define the \textit{ladder set} $L_{u,v}$, a subset of $\sigma_{u,v}$ which provides the hypercube you have to go through if you want to reach $v$ from $u$ with a shortest path. We see $u$ as the departure vertex, while $v$ is the arrival. The notion of ladder set is defined only for vertices $u,v$ satisfying $u \in I(v_0,v)$. Another characterization involving two ladder sets will follow for any pair of vertices.

\begin{definition}[Ladder set $L_{u,v}$]
Let $u,v \in V$ such that $u \in I(v_0,v)$. The ladder set $L_{u,v}$ is the subset of $\sigma_{u,v}$ which contains the classes admitting an edge adjacent to $u$.
\label{def:ladder}
\end{definition}

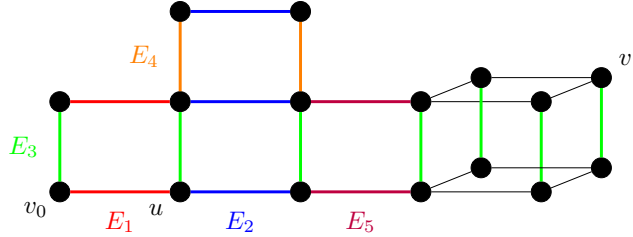
\begin{figure}[h]
\centering
\scalebox{0.8}{\begin{tikzpicture}


\node[draw, circle, minimum height=0.2cm, minimum width=0.2cm, fill=black] (P11) at (1,1) {};
\node[draw, circle, minimum height=0.2cm, minimum width=0.2cm, fill=black] (P12) at (1,2.5) {};

\node[draw, circle, minimum height=0.2cm, minimum width=0.2cm, fill=black] (P21) at (3,1) {};
\node[draw, circle, minimum height=0.2cm, minimum width=0.2cm, fill=black] (P22) at (3,2.5) {};
\node[draw, circle, minimum height=0.2cm, minimum width=0.2cm, fill=black] (P23) at (3,4) {};

\node[draw, circle, minimum height=0.2cm, minimum width=0.2cm, fill=black] (P31) at (5,1) {};
\node[draw, circle, minimum height=0.2cm, minimum width=0.2cm, fill=black] (P32) at (5,2.5) {};
\node[draw, circle, minimum height=0.2cm, minimum width=0.2cm, fill=black] (P33) at (5,4) {};

\node[draw, circle, minimum height=0.2cm, minimum width=0.2cm, fill=black] (P41) at (7,1) {};
\node[draw, circle, minimum height=0.2cm, minimum width=0.2cm, fill=black] (P42) at (7,2.5) {};

\node[draw, circle, minimum height=0.2cm, minimum width=0.2cm, fill=black] (P51) at (9,1) {};
\node[draw, circle, minimum height=0.2cm, minimum width=0.2cm, fill=black] (P52) at (9,2.5) {};

\node[draw, circle, minimum height=0.2cm, minimum width=0.2cm, fill=black] (P61) at (8.0,1.4) {};
\node[draw, circle, minimum height=0.2cm, minimum width=0.2cm, fill=black] (P62) at (8.0,2.9) {};
\node[draw, circle, minimum height=0.2cm, minimum width=0.2cm, fill=black] (P63) at (10.0,1.4) {};
\node[draw, circle, minimum height=0.2cm, minimum width=0.2cm, fill=black] (P64) at (10.0,2.9) {};


\draw[line width = 1.4pt, color=green] (P11) -- (P12);
\draw[line width = 1.4pt, color=red] (P11) -- (P21);
\draw[line width = 1.4pt, color=red] (P12) -- (P22);
\draw[line width = 1.4pt, color=green] (P21) -- (P22);

\draw[line width = 1.4pt, color=blue] (P21) -- (P31);
\draw[line width = 1.4pt, color=blue] (P22) -- (P32);
\draw[line width = 1.4pt, color=green] (P31) -- (P32);

\draw[line width = 1.4pt, color=orange] (P22) -- (P23);
\draw[line width = 1.4pt, color=blue] (P23) -- (P33);
\draw[line width = 1.4pt, color=orange] (P32) -- (P33);

\draw[line width = 1.4pt, color=purple] (P31) -- (P41);
\draw[line width = 1.4pt, color=purple] (P32) -- (P42);
\draw[line width = 1.4pt, color=green] (P41) -- (P42);

\draw (P41) -- (P51);
\draw (P42) -- (P52);
\draw[line width = 1.4pt, color=green] (P51) -- (P52);

\draw (P41) -- (P61);
\draw (P42) -- (P62);
\draw (P51) -- (P63);
\draw (P52) -- (P64);
\draw[line width = 1.4pt, color=green] (P61) -- (P62);
\draw (P61) -- (P63);
\draw (P62) -- (P64);
\draw[line width = 1.4pt, color=green] (P63) -- (P64);


\node[scale=1.2, color = red] at (2.0,0.5) {$E_1$};
\node[scale=1.2, color = blue] at (4.0,0.5) {$E_2$};
\node[scale=1.2, color = green] at (0.4,1.75) {$E_3$};
\node[scale=1.2, color = orange] at (2.4,3.25) {$E_4$};
\node[scale=1.2, color = purple] at (6.0,0.5) {$E_5$};

\node[scale = 1.2] at (0.6,0.7) {$v_0$};
\node[scale = 1.2] at (2.6,0.7) {$u$};
\node[scale = 1.2] at (10.4,3.2) {$v$};

\end{tikzpicture}}
\caption{An example with $u \in I(v_0,v)$ and the ladder set $L_{u,v} = \set{E_2,E_3}$}
\label{fig:ladder_sets}
\end{figure}

Observe that when $u \in I(v_0,v)$, all edges of any shortest $(u,v)$-path are oriented towards $v$ with the $v_0$-orientation. The definition of the ladder set takes into account an order of the pair $(u,v)$ while it is not the case for the signature. Figure~\ref{fig:ladder_sets} shows a median graph $G$ and a pair $u,v$ of vertices satisfying $u \in I(v_0,v)$. We color certain $\Theta$-classes of this graph. The $\Theta$-classes belonging to $\sigma_{u,v}$ and which are adjacent to $u$ are $E_2$ and $E_3$: $L_{u,v} = \set{E_2,E_3}$. For example, $E_5 \in \sigma_{u,v}$ but there is no edge of $E_5$ adjacent to $u$. We show that any ladder set is a POF.

\begin{theorem}
Suppose $u \in I(v_0,v)$. Any ladder set $L_{u,v}$ is a POF. Moreover, there exists a shortest $(u,v)$-path, where its first edges belong to $L_{u,v}$.
\label{th:ladder_pof}
\end{theorem}
\begin{proof}
If $\card{L_{u,v}} = 1$, then the proof is terminated. Suppose $\card{L_{u,v}} \ge 2$. We denote by $E_i$ and $E_j$ two arbitrary classes in $L_{u,v}$.  We prove that they are orthogonal. 

Let $e_i = (u,z_i)$ be the edge of $E_i$ adjacent to $u$. The one of class $E_j$ is denoted by $e_j = (u,z_j)$. As $z_i$ is the gate of $u$ in $H_i''$, we have $d(z_i,v) = d(u,v)-1$. Moreover, $\sigma_{z_i,v} = \sigma_{u,v}\backslash \set{E_i}$. 
Any shortest $(z_i,v)$-path concatenated with edge $(u,z_i)$ produces a shortest $(u,v)$-path. Furthermore, it contains an edge of $E_j$, so one of its vertex is in $\partial H_j'$. In summary, we know that set $\partial H_j'$ is convex (Lemma~\ref{le:boundaries}) and contains both $u$ and a vertex of any shortest $(z_i,v)$-path. Consequently, $z_i \in \partial H_j'$: there is an edge $(z_i,w) \in E_j$. As the edges of $E_j$ define an isomorphism (Lemma~\ref{le:boundaries}), $(w,z_j) \in E$. We obtain a square made up of classes $E_i$ and $E_j$, so $E_i \perp E_j$. As it is true for two arbitrary classes $E_i,E_j \in L_{u,v}$, set $L_{u,v}$ is a POF.

According to Lemma~\ref{le:pof_adjacent}, there is an hypercube $Q$ containing $u$ and with signature $L_{u,v}$. As $u \in \partial H_i'$ for any $E_i \in L_{u,v}$, vertex $u$ is the basis of hypercube $Q$.
Let $u^+$ be the opposite vertex of $u$ in $Q$. Consider the concatenation of a shortest $(u,u^+)$-path $P$ made up of edges in $Q$ and of a shortest $(u^+,v)$-path $P^+$. The signature of $P$ is $L_{u,v}$. As $L_{u,v} \subseteq \sigma_{u,v}$, vertex $u^+$ is in $I(u,v)$. Consequently, the concatenation $P \cdot P^+$ is a shortest $(u,v)$-path and its first edges belong to the $\Theta$-classes in $L_{u,v}$.
\end{proof}

More generally, for any ordering $\tau$ of the classes of $L_{u,v}$, we can identify a shortest $(u,v)$-path, where its first part is made up with edges  of the classes of $L_{u,v}$ following the ordering $\tau$. This prefix is in fact a shortest path between $u$ and its opposite vertex in the hypercube containing $u$ with signature $L_{u,v}$. However, several shortest $(u,v)$-paths may contain classes of $\sigma_{u,v} \backslash L_{u,v}$ before certain classes of $L_{u,v}$.

The main principle of our algorithm is to label each pair made up of a vertex $u \in V$ and a POF $L$ ``outgoing'' from $u$, {\em i.e.} such that, for each $E_i \in L$, there is an edge of $E_i$ outgoing from $u$. We know from Lemma~\ref{le:pof_adjacent} that it means there is an hypercube made up of edges in $L$ and with basis $u$. Then, the label associated with $u$ and $L$ is the maximum distance $d(u,v)$ for a pair $(u,v)$, $u \in I(v_0,v)$, admitting $L$ as its ladder set. There is one label for each hypercube of $G$.

We showed how the ladder sets characterize pairs of vertices $u,v$, where $u$, $v$, and $v_0$ are aligned, {\em i.e.} $u \in I(v_0,v)$. We focus now on the general case: we can characterize any pair $u,v$ of vertices with two ladder sets.

\begin{definition}[Ladder pair $(L_{m,u},L_{m,v})$]
Let $m = m(u,v,v_0)$. As $m \in I(v_0,u)$ and $m \in I(v_0,v)$, pairs $(m,u)$ and $(m,v)$ admit a ladder set. The ladder pair of $u$ and $v$ is $(L_{m,u},L_{m,v})$.
\label{def:ladder_pair}
\end{definition}

If $u \in I(v_0,v)$, then $m = u$. In this case, the ladder pair is the empty set with the ladder set of $u$ and $v$. Figure~\ref{fig:ladder_pairs} represents a triplet $u,m,v$ such that $m = m(u,v,v_0)$. The edges are oriented according to the $v_0$-orientation. In this example, the ladder pair of $u,v$ is $(L_{m,u},L_{m,v})$, with $L_{m,u} = \set{E_i,E_j}$ and $L_{m,v} = \set{E_r,E_{\ell}}$.

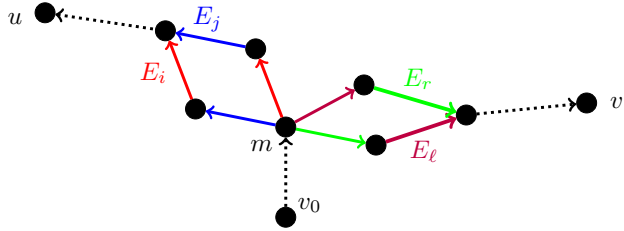
\begin{figure}[h]
\centering
\scalebox{0.8}{\begin{tikzpicture}


%
%
\node[draw, circle, minimum height=0.2cm, minimum width=0.2cm, fill=black] (P3) at (3.0,4.4) {};

\node[draw, circle, minimum height=0.2cm, minimum width=0.2cm, fill=black] (P4) at (5.5,2.8) {};
\node[draw, circle, minimum height=0.2cm, minimum width=0.2cm, fill=black] (P42) at (5.0,4.1) {};

\node[draw, circle, minimum height=0.2cm, minimum width=0.2cm, fill=black] (P5) at (7,2.5) {};
\node[draw, circle, minimum height=0.2cm, minimum width=0.2cm, fill=black] (P52) at (6.5,3.8) {};
\node[draw, circle, minimum height=0.2cm, minimum width=0.2cm, fill=black] (v0) at (7,1.0) {};

\node[draw, circle, minimum height=0.2cm, minimum width=0.2cm, fill=black] (P6) at (8.5,2.2) {};
\node[draw, circle, minimum height=0.2cm, minimum width=0.2cm, fill=black] (P62) at (8.3,3.2) {};

\node[draw, circle, minimum height=0.2cm, minimum width=0.2cm, fill=black] (P7) at (10.0,2.7) {};
\node[draw, circle, minimum height=0.2cm, minimum width=0.2cm, fill=black] (P8) at (12.0,2.9) {};


\draw[->,dotted,line width = 1.4pt] (v0) -- (P5);

\draw[<-,line width = 1.4pt, dotted] (P3) -- (P42);
\draw[<-,line width = 1.4pt,color = blue] (P4) -- (P5);
\draw[->,line width = 1.4pt, color = green] (P5) -- (P6);
\draw[->,line width = 1.4pt, color = purple] (P5) -- (P62);
\draw[->,line width = 1.8pt, color = purple] (P6) -- (P7);
\draw[->,line width = 1.8pt, color = green] (P62) -- (P7);
\draw[->,line width = 1.4pt, dotted] (P7) -- (P8);

\draw[->,line width = 1.4pt, color = red] (P4) -- (P42);
\draw[->,line width = 1.4pt, color = red] (P5) -- (P52);
\draw[->,line width = 1.4pt, color = blue] (P52) -- (P42);


\node[scale=1.2, color = blue] at (5.7,4.3) {$E_j$};
\node[scale=1.2, color = red] at (4.8,3.4) {$E_i$};

\node[scale=1.2, color = green] at (9.2,3.3) {$E_r$};
\node[scale=1.2, color = purple] at (9.3,2.1) {$E_{\ell}$};

\node[scale = 1.2] at (7.4,1.2) {$v_0$};
\node[scale = 1.2] at (6.6,2.2) {$m$};

\node[scale = 1.2] at (2.5,4.3) {$u$};
\node[scale = 1.2] at (12.5,2.9) {$v$};

\end{tikzpicture}}
\caption{Ladder pairs $L_{m,u} = \set{E_i,E_j}$ and $L_{m,v} = \set{E_r,E_{\ell}}$.}
\label{fig:ladder_pairs}
\end{figure}

A shortest $(u,v)$-path passing through $m = m(u,v,v_0)$ is refered as an ``down-up'' shortest path in~\cite{KlMu99}. Indeed, its $(u,m)$-section is $v_0$-oriented ``downwards'' - edges are oriented towards $u$ - while its $(m,v)$-section is oriented ``upwards'' - edges are oriented towards $v$. The ladder pair of $(u,v)$ indicates the ladder set of both the downwards and the upwards section. They contain the $\Theta$-classes adjacent to $m$ which belong to $\sigma_{m,u}$ and $\sigma_{m,v}$ respectively.

\textbf{Computation of the labels.} We design an algorithm which computes the labels on our median graph $G$. It determines for each pair $(u,L)$, $u \in V$ and $L$ is a POF outgoing from $u$, the length of the longest shortest path starting from $u$, with an arrival $v \in V$ verifying $u \in I(v_0,v)$, and such that its ladder set is $L$. We denote by $\varphi(u,L)$ this variable. Put formally, $\varphi(u,L)$ is the maximum distance $d(u,v)$, where $u \in I(v_0,v)$ and $L_{u,v} = L$.

The following theorem identifies a relationship between labels of different vertices.

\begin{theorem}
Let $u \in V$, $L$ be a POF outgoing from $u$ and $Q$ be the hypercube with basis $u$ and signature $L$. We denote by $u^+$ the opposite vertex of $u$ in $Q$: $u$ is the basis of $Q$ and $u^+$ its anti-basis. A vertex $v\neq u^+$ verifies $u \in I(v_0,v)$ and $L_{u,v} = L$ if and only if (i) $u^+ \in I(v_0,v)$ and (ii) for each $E_j \in L_{u^+,v}$, $L \cup \set{E_j}$ is not a POF.
\label{th:pushing_label}
\end{theorem}
\begin{proof}
Let $P_{uu^+}$ be a shortest $(u,u^+)$-path made up of edges of the hypercube $Q$. Let $P_{u^+v}$ be a shortest $(u^+,v)$-path. All edges of path $P_{uu^+}$ are oriented towards $u^+$. Similarly, path $P_{u^+v}$ is oriented towards $v$ as $u^+ \in I(v_0,v)$. Consequently, the concatenation $P_{uu^+} \cdot P_{u^+v}$ is an $(u,v)$-path whose edges are all oriented towards $v$. According to Theorem~\ref{th:at_most_one}, it is a shortest path, otherwise it would pass twice through the same class and an edge of the path would be oriented towards $u$. So, $u^+ \in I(u,v)$.

We can use the same argument with the triplet $(v_0,u,v)$, considering a shortest $(v_0,u)$-path and the shortest $(u,v)$-path $P_{uu^+} \cdot P_{u^+v}$ which is oriented towards $v$. We obtain that $u \in I(v_0,v)$.

Finally, we prove that the ladder set of the pair $(u,v)$ is exactly $L$. Suppose that a class $E_j$ of $L_{u^+,v}$ contains an edge adjacent to $u$. As $u,u^+ \in \partial H_j'$, the convexity of boundaries (Lemma~\ref{le:boundaries}) implies that all vertices of $Q$ are in $\partial H_j'$ because $Q \subseteq I(u,u^+)$. The class $E_j$ defines an isomorphism, so $E_j$ is orthogonal to all classes represented in $Q$, which corresponds to $L$. In other words, $L \cup \set{E_j}$ is a POF, a contradiction.

The converse is also true. Let us consider a vertex $v$ satisfying $u \in I(v_0,v)$ and $L_{uv} = L$. We have $u^+ \in I(u,v) \subseteq I(v_0,v)$, therefore the pair $(u^+,v)$ admits a ladder set. Suppose by way a contradiction that there is $E_j \in L_{u^+,v}$ such that $L \cup \set{E_j}$ is a POF. We know that $u^+ \in \partial H_j'$. Lemma~\ref{le:squares} implies that the neighbors of $u^+$ in $Q$ are also in $\partial H_j'$. We can spread this reasoning: the vertices at distance 2 from $u^+$ in $Q$ are in $\partial H_j'$, etc. Finally, we have $u \in \partial H_j'$. This is a contradiction as the ladder set $L_{uv}$ is $L$, not $L \cup \set{E_j}$.
\end{proof}

\begin{algorithm}[t]
\SetKwFor{For}{for}{do}{\nl endfor}
\SetKwFor{Forall}{for all}{do}{\nl endfor}
\SetKwIF{If}{ElseIf}{Else}{if}{then}{else if}{else}{}
\DontPrintSemicolon
\SetNlSty{}{}{:}
\SetAlgoNlRelativeSize{0}
\SetNlSkip{1em}
\nl\KwIn{graph $G$, $\Theta$-classes $\mathcal{E}$, list $\mathcal{Q}$ of hypercubes (triplets anti-basis, basis, and POF from Lemma~\ref{le:enum_hypercubes})}
\nl\KwOut{Labels $\varphi(u,L)$ for each vertex $u\in V$ and POF $L$ outgoing from $u$}
\nl Initialize $\varphi(u,L) \leftarrow 0$ for each $u \in V$ and $L$ POF outgoing from $u$;\;
\nl $\mathcal{Q}^* \leftarrow \mbox{\textbf{reverse}}(\mathcal{Q})$; \label{line:reverse}\; 
\nl \For{$(u^+,u,L)$ in list $\mathcal{Q}^*$}{
	\nl \If{$\varphi(u,L) = 0$}{
		\nl $\varphi(u,L) \leftarrow \card{L}$; \label{line:phi_zero}\;
	}
	\ifend\;
	\nl \Forall{$X \subseteq \mathcal{E}^-(u), X \neq \emptyset$}{
		\nl $u^- \leftarrow$ basis of the hypercube with anti-basis $u$ and $\Theta$-classes $X$; \label{line:u_moins}\;
		\nl $\checkperp \leftarrow \mbox{false}$; \label{line:init_check}\;
		\nl \Forall{$E_j \in L$}{
			\nl \lIf{$X \cup \set{E_j}$ is a POF}{$\checkperp \leftarrow \mbox{true}$;\label{line:check}}}
		\nl \lIf{{\em not} $\checkperp$}{$\varphi(u^-,X) \leftarrow \max \set{\varphi(u^-,X), \card{X}+\varphi(u,L)}$\label{line:update_phi};}} 
	}	

\caption{The computation of labels $\varphi(u,L)$}
\label{algo:labels}
\end{algorithm}

Theorem~\ref{th:pushing_label} is illustrated in Figure~\ref{fig:compute_labels}. It shows a median graph with a basepoint $v_0$. Some $\Theta$-classes are colored. Consider the POF $L =\set{E_2,E_3}$ outgoing from $u$ and the ladder set $L_{u^+,x} = \set{E_1}$. As $L \cup \set{E_1}$ is a POF, the ladder set $L_{u,x}$ of pair $u,x$ is not $L$, but $L_{u,x} = \set{E_1,E_2,E_3}$. However, if we consider the ladder set $L_{u^+,v} = \set{E_4,E_5}$, then we observe that neither $L \cup \set{E_4}$ nor $L \cup \set{E_5}$ is a POF. So, $L_{u,v} = L$ according to Theorem~\ref{th:pushing_label}. Indeed, one can check that $E_2$ and $E_3$ are the only $\Theta$-classes of $\sigma_{u,v}$ which are adjacent to $u$.

\begin{figure}[h]
\centering
\scalebox{0.8}{\begin{tikzpicture}


\node[draw, circle, minimum height=0.2cm, minimum width=0.2cm, fill=black] (P01) at (-1,1) {};

\node[draw, circle, minimum height=0.2cm, minimum width=0.2cm, fill=black] (P11) at (1,1) {};
\node[draw, circle, minimum height=0.2cm, minimum width=0.2cm, fill=black] (P12) at (1,2.5) {};
\node[draw, circle, minimum height=0.2cm, minimum width=0.2cm, fill=black] (P13) at (3,1) {};
\node[draw, circle, minimum height=0.2cm, minimum width=0.2cm, fill=black] (P14) at (3,2.5) {};

\node[draw, circle, minimum height=0.2cm, minimum width=0.2cm, fill=black] (P21) at (2.0,1.4) {};
\node[draw, circle, minimum height=0.2cm, minimum width=0.2cm, fill=black] (P22) at (2.0,2.9) {};
\node[draw, circle, minimum height=0.2cm, minimum width=0.2cm, fill=black] (P23) at (4.0,1.4) {};
\node[draw, circle, minimum height=0.2cm, minimum width=0.2cm, fill=black] (P24) at (4.0,2.9) {};

\node[draw, circle, minimum height=0.2cm, minimum width=0.2cm, fill=black] (P31) at (3.0,4) {};
\node[draw, circle, minimum height=0.2cm, minimum width=0.2cm, fill=black] (P32) at (5.0,4) {};

\node[draw, circle, minimum height=0.2cm, minimum width=0.2cm, fill=black] (P41) at (5.0,1) {};
\node[draw, circle, minimum height=0.2cm, minimum width=0.2cm, fill=black] (P42) at (5.0,2.5) {};

\node[draw, circle, minimum height=0.2cm, minimum width=0.2cm, fill=black] (P51) at (8.0,5.0) {};
\node[draw, circle, minimum height=0.2cm, minimum width=0.2cm, fill=black] (P52) at (4.5,5.5) {};


\draw[line width = 1.4pt] (P01) -- (P11);

\draw[line width = 1.4pt, color = green] (P11) -- (P12);
\draw[line width = 1.4pt, color = blue] (P11) -- (P13);
\draw[line width = 1.4pt, color = red] (P11) -- (P21);
\draw[line width = 1.4pt, color = blue] (P12) -- (P14);
\draw[line width = 1.4pt, color = red] (P12) -- (P22);
\draw[line width = 1.4pt, color = red]  (P13) -- (P23);
\draw[line width = 1.4pt, color = green] (P13) -- (P14);
\draw[line width = 1.4pt, color = red]  (P14) -- (P24);
\draw[line width = 1.4pt, color = green] (P21) -- (P22);
\draw[line width = 1.4pt, color = blue] (P21) -- (P23);
\draw[line width = 1.4pt, color = green] (P23) -- (P24);
\draw[line width = 1.4pt, color = blue] (P22) -- (P24);

\draw[line width = 1.4pt, color = purple] (P14) -- (P31);
\draw[line width = 1.4pt, color = purple] (P42) -- (P32);
\draw[line width = 1.4pt, color = orange] (P13) -- (P41);
\draw[line width = 1.4pt, color = orange] (P14) -- (P42);
\draw[line width = 1.4pt, color = orange] (P31) -- (P32);
\draw[line width = 1.4pt, color = green] (P41) -- (P42);

\draw[dashed] (P32) -- (P51);
\draw[dashed] (P24) -- (P52);


\node[scale=1.2, color = red] at (1.4,3.0) {$E_1$};
\node[scale=1.2, color = blue] at (2.0,0.5) {$E_2$};
\node[scale=1.2, color = green] at (0.5,1.75) {$E_3$};
\node[scale=1.2, color = orange] at (4.0,0.5) {$E_4$};
\node[scale=1.2, color = purple] at (5.5,3.25) {$E_5$};

\node[scale = 1.2] at (-1.4,0.7) {$v_0$};
\node[scale = 1.2] at (0.6,0.7) {$u$};
\node[scale = 1.2] at (2.6,2.2) {$u^+$};
\node[scale = 1.2] at (7.6,5.2) {$v$};
\node[scale = 1.2] at (4.1,5.3) {$x$};

\end{tikzpicture}}
\caption{Two vertices $u$ and $u^+$ respectively basis and anti-basis of an hypercube with classes $L = \set{E_2,E_3}$. POF $L$ is the ladder set of $(u,v)$, not of $(u,x)$. Indeed, $L_{u,x} = \set{E_1,E_2,E_3}$.}
\label{fig:compute_labels}
\end{figure}
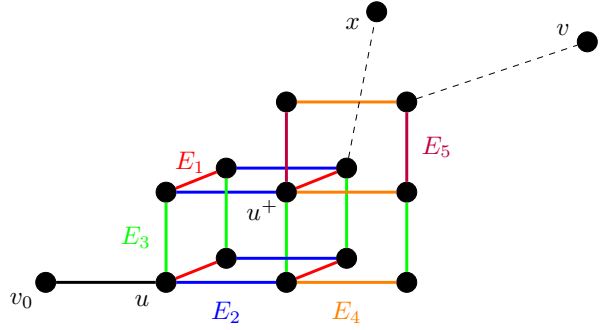

We initialize all labels $\varphi(u,L)$ with 0. We consider the list of hypercubes obtained from Lemma~\ref{le:enum_hypercubes}. We go through this list in the descending order, {\em i.e.} starting with hypercubes having the farthest-to-$v_0$ anti-basis. For each triplet treated, we pick up two of the three elements: the basis $u$ and the sets of $\Theta$-classes of the hypercube that we denote by $L$. Set $L$ is indeed a POF outgoing from $u$. Suppose that we are treating the pair $(u,L)$ corresponding to the hypercube $Q$: all pairs $(u^+,L^+)$, where $u^+$ is the anti-basis of $Q$, have already been considered. Indeed, pairs $(u^+,L^+)$ admit an anti-basis which is farthest to $v_0$ compared to the anti-basis $u^+$ of $(u,L)$.

We give the details of our procedure. Algorithm~\ref{algo:labels} presents it in pseudocode. First, we verify the value we have for variable $\varphi(u,L)$. If $\varphi(u,L) = 0$, then we put $\varphi(u,L) \leftarrow \card{L}$ (line~\ref{line:phi_zero} in Algorithm~\ref{algo:labels}). This case occurs when the hypercube $Q$ with basis $u$ and POF $L$ is ``peripheral'', {\em i.e.} when the anti-basis $u^+$ of $Q$ is the vertex $v$ which is the farthest to $u$ and such that $L_{uv} = L$. The distance $d(u,u^+)$ is equal to the dimension of hypercube $Q$, which is $\card{L}$. If $\varphi(u,L) > 0$, then we consider that the label provides us with the distance from $u$ to the farthest-to-$u$ vertex satisfying $L_{uv} = L$. Indeed, the anti-basis $u^+$ of $Q$ has already been treated by the algorithm. We know from Theorem~\ref{th:pushing_label} that the value $\varphi(u,L)$ can be deduced from a pair $(u^+,L^+)$, where $L^+$ is a POF outgoing from $u^+$ which does not contain a class forming a POF with $L$.

Second, we list all subsets of $\mathcal{E}^-(u)$. The idea is to communicate the value $\varphi(u,L)$ to certain vertices which are the bases of hypercubes with anti-basis $u$. Let us consider one subset, $X \subseteq \mathcal{E}^-(u)$.  We know from Lemma~\ref{le:pof_hypercube} that there is an hypercube $Q^-$ with anti-basis $u$ formed by the classes $X$. We denote by $u^-$ the basis of the hypercube $Q^-$ (line~\ref{line:u_moins}). We verify whether there is a class $E_j \in L$ such that $X \cup \set{E_j}$ is a POF (line~\ref{line:check}). If the answer is no, then there is a vertex $v$ such that $L_{u^-v} = X$ which is at distance $\card{X}+\varphi(u,L)$ from $u$, according to Theorem~\ref{th:pushing_label}. We update variable $\varphi(u^-,X)$ (line~\ref{line:update_phi}):
\[\varphi(u^-,X) \leftarrow \max \set{\varphi(u^-,X), \card{X}+\varphi(u,L)}.
\]

\begin{theorem}
The execution of Algorithm~\ref{algo:labels} produces labels $\varphi(u,L)$ which are the maximum distances $d(u,v)$, where $v$ satisfying $L_{u,v} = L$.
\label{th:algo_labels_works}
\end{theorem}
\begin{proof}
The algorithm treats all pairs $(u,L)$ such that $u$ is a vertex and $L$ is a POF outgoing from $u$. The first action of the algorithm is to check the current value of $\varphi(u,L)$.

If $\varphi(u,L) = 0$, it means that there is no edge outgoing from $u^+$, the anti-basis of the hypercube with basis $u$ and signature $L$. In this case, $u^+$ is the farthest-to-$u$ vertex with ladder set $L$. So, $\varphi(u,L) = d(u,u^+) = \card{L}$.

If $\varphi(u,L) > 0$, there is an hypercube with basis $u^+$ and signature $L^+$, where $L \cup \set{E_i}$ is not a POF for any $E_i \in L^+$. According to Theorem~\ref{th:pushing_label}, the algorithm ensures us that the label of $(u,L)$ is correctly computed as we consider all pairs $(u^+,L^+)$ susceptible to influence value $\varphi(u,L)$.
\end{proof}

We focus now on the running time of Algorithm~\ref{algo:labels}. For each pair $(u,L)$ (said differently for any hypercube of $G$), we list all subsets of set $\mathcal{E}^-(u)$ which is of size at most $d$ (Lemma~\ref{le:pof_hypercube}). For each of these subsets $X$, we verify whether $X \cup \set{E_j}$ is a POF, where $E_j \in L$. The number of hypercubes is upper-bounded by $2^dn$ (Lemma~\ref{le:number_hypercubes}) and the number of subsets $X$ is at most $2^d$. Determining whether $X \cup \set{E_j}$ is a POF can be done in $O(1)$ by checking whether vertex $u^-$ is adjacent to an edge of $E_j$. This argument comes from the convexity of boundaries, already used twice in the proof of Theorem~\ref{th:pushing_label}.  So, the total running time is $O(d2^{2d}n)$.

\subsection{Maximum-weighted disjoint sets} \label{subsec:mwds}

In this section, we use the labels computed in Section~\ref{subsec:labels} to determine the diameter $\diam(G)$ of the median graph $G$. The idea is to explore ladder pairs $(L_{m,u},L_{m,v})$ for all vertices $m\in V$ and to find the one that maximizes the distance from a vertex $u$ to a vertex $v$ which satisfy $m = m(u,v,v_0)$.

\textbf{A relationship between the diameter and ladder pairs.} We begin with a theorem which allows us to characterize the diametral pair regarding the ladder pairs.

\begin{theorem}
Let $m \in V$ and $L,L^*$ be two POFs outgoing from $m$. Let $u,v$ be two vertices such that $m$ belong to both $I(v_0,u)$ and $I(v_0,v)$. Suppose they fulfil $L_{m,u} = L$ and $L_{m,v} = L^*$. Then, $m \in I(u,v)$ if and only if $L \cap L^* = \emptyset$.
\label{th:disjoint_POF}
\end{theorem}
\begin{proof}
First, suppose that $L \cap L^* \neq \emptyset$. Then $\sigma_{m,u} \cap \sigma_{m,v} \neq \emptyset$. As a consequence, $m \notin I(u,v)$: the concatenation of a shortest $(u,m)$-path with a shortest $(m,v)$-path does not produce a shortest $(u,v)$-path because it contains two edges from the same class (Theorem~\ref{th:signature}).

Second, suppose that $L \cap L^* = \emptyset$. We prove that $m \in I(u,v)$. Let $P_{m,u}$ (resp. $P_{m,v}$) be an arbitrary shortest $(m,u)$-path (resp. $(m,v)$-path). If $\sigma_{m,u} \cap \sigma_{m,v} = \emptyset$, then the concatenation of $P_{m,u}$ and $P_{m,v}$ produces a shortest $(u,v$)-path traversing $m$, so $m \in I(u,v)$. This is why we suppose, by way of contradiction, that there is a class $E_{i}$ appearing both in $P_{m,u}$ and $P_{m,v}$. There is one edge $(u_i,u_i^*)$ of $E_{i}$ in $P_{m,u}$ and one edge $(v_i,v_i^*)$ of $E_{i}$ in $P_{m,v}$. This notation will be used for any class belonging to both paths, {\em i.e.} in $\sigma_{m,u} \cap \sigma_{m,v}$.

Instead of considering an arbitrary class $E_i$ admitting one edge in both paths, we denote by $E_j$ a class of $\sigma_{m,u} \cap \sigma_{m,v}$ satisfying a certain property: there is no other class $E_h \in \sigma_{m,u} \cap \sigma_{m,v}$ with both $(u_h,u_h^*)$ closer to $m$ than $(u_j,u_j^*)$ in $P_{m,u}$ and $(v_h,v_h^*)$ closer to $m$ than $(v_j,v_j^*)$ in $P_{m,v}$. Said differently, we fix $E_j$ such that another class of $\sigma_{m,u} \cap \sigma_{m,v}$ that would be closer to $m$ in both paths does not exist. Obviously, such a class $E_j$ exists as set $\sigma_{m,u} \cap \sigma_{m,v}$ is finite.

\begin{figure}[h]
\centering
\scalebox{0.8}{\begin{tikzpicture}


\node[draw, circle, minimum height=0.2cm, minimum width=0.2cm, fill=black] (P1) at (1,4) {};

\node[draw, circle, minimum height=0.2cm, minimum width=0.2cm, fill=black] (P2) at (2.5,3.5) {};

\node[draw, circle, minimum height=0.2cm, minimum width=0.2cm, fill=red] (P3) at (4,3.1) {};

\node[draw, circle, minimum height=0.2cm, minimum width=0.2cm, fill=red] (P4) at (5.5,2.8) {};
\node[draw, circle, minimum height=0.2cm, minimum width=0.2cm, fill=black] (P42) at (5.0,4.1) {};

\node[draw, circle, minimum height=0.2cm, minimum width=0.2cm, fill=red] (P5) at (7,2.5) {};
\node[draw, circle, minimum height=0.2cm, minimum width=0.2cm, fill=black] (P52) at (6.5,3.8) {};
\node[draw, circle, minimum height=0.2cm, minimum width=0.2cm, fill=black] (v0) at (7,1.0) {};

\node[draw, circle, minimum height=0.2cm, minimum width=0.2cm, fill=red] (P6) at (8.5,2.8) {};

\node[draw, circle, minimum height=0.2cm, minimum width=0.2cm, fill=black] (P7) at (10.0,3.1) {};

\node[draw, circle, minimum height=0.2cm, minimum width=0.2cm, fill=black] (P8) at (11.5,3.5) {};


\draw[->,dotted,line width = 1.4pt] (v0) -- (P5);

\draw[<-,line width = 1.4pt] (P1) -- (P2);
\draw[<-,line width = 1.8pt, color = blue] (P2) -- (P3);
\draw[<-,line width = 1.4pt] (P3) -- (P4);
\draw[<-,line width = 1.4pt] (P4) -- (P5);
\draw[->,line width = 1.4pt] (P5) -- (P6);
\draw[->,line width = 1.8pt, color = blue] (P6) -- (P7);
\draw[->,line width = 1.4pt] (P7) -- (P8);

\draw[->,line width = 1.8pt, color = blue] (P4) -- (P42);
\draw[->,line width = 1.8pt, color = blue] (P5) -- (P52);


\node[scale=1.2, color = blue] at (3.3,3.6) {$E_j$};
\node[scale=1.2, color = red] at (5.8,2.1) {$\partial H_j'$};

\node[scale = 1.2] at (7.4,1.2) {$v_0$};
\node[scale = 1.2] at (7.3,2.8) {$m$};
\node[scale = 1.2, color = red] at (4.0,2.7) {$u_j$};
\node[scale = 1.2, color = red] at (8.5,2.4) {$v_j$};
\node[scale = 1.2] at (2.5,3.0) {$u_j^*$};
\node[scale = 1.2] at (10.0,2.6) {$v_j^*$};

\end{tikzpicture}}
\caption{Illustration of the convexity of $\partial H_j'$}
\label{fig:convexity_boundary}
\end{figure}
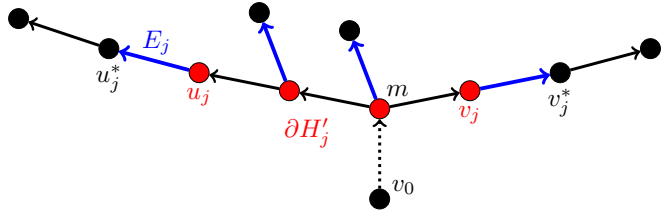

Let $u_j$ (resp. $v_j$) be the vertex in $\partial H_j'$ of path $P_{m,u}$ (resp. $P_{m,v}$). In other words, $u_j$ is the vertex adjacent to the edge of $E_j$ in $P_{m,u}$ which is the closest to $v_0$, so the closest to $m$. With the definition of class $E_j$, we have $m \in I(u_j,v_j)$. Indeed, the concatenation of the $(m,u_j)$-section of $P_{m,u}$ and of the $(m,v_j)$-section of $P_{m,v}$ produces a shortest $(u_j,v_j)$-path passing through $m$ as it does not contain two edges of the same class. As $u_j,v_j \in \partial H_j'$ and $\partial H_j'$ is convex, then vertex $m$ is adjacent to an edge of $E_j$: $m \in \partial H_j'$. Figure~\ref{fig:convexity_boundary} illustrates the latter assertion: the edges of $E_j$ are drawn in blue, the vertices of $\partial H_j'$ in red. As $E_j \in \sigma_{m,u} \cap \sigma_{m,v}$, we have $E_j \in L_{m,u} \cap L_{m,v}$ which is a contradiction.
\end{proof}

Thanks to Theorem~\ref{th:disjoint_POF}, we establish a relationship between the diameter $\diam(G)$ and the ladder pairs. For $m \in V$, let $\Upsilon(m)$ be the distance of the longest shortest $(u,v)$-path such that $m = m(u,v,v_0)$. The maximum $\Upsilon(G)$ over all $\Upsilon(m)$, $m \in V$ is naturally: $\Upsilon(G) = \diam(G)$.

\begin{corollary} For any $m \in V$, value $\Upsilon(m)$ is the maximum of $\varphi(m,L) + \varphi(m,L^*)$ for two POFs $L$ and $L^*$ outgoing from $m$ and with an empty intersection. Formally,
\begin{equation}
\Upsilon(m) = \max_{\substack{ L \cap L^* = \emptyset \\  \mbox{{\em \scriptsize{outgoing from}}}~ m}} \varphi(m,L) + \varphi(m,L^*).
\label{eq:upsilon_m}
\end{equation}

The diameter of $G$ can thus be written:
\begin{equation}
\diam(G) = \Upsilon (G) = \max_{m \in V} \max_{\substack{ L \cap L^* = \emptyset \\ \mbox{{\em \scriptsize{outgoing from}}}~ m}} \varphi(m,L) + \varphi(m,L^*).
\label{eq:upsilon_G}
\end{equation}
\label{co:diameter_sets}
\end{corollary}
\begin{proof}
Let $m \in V$ and $u,v$ be the pair of vertices satisfying $m = m(u,v,v_0)$ which maximizes $d(u,v)$. We know from Theorem~\ref{th:disjoint_POF} that the ladder sets $L_{m,u}$ and $L_{m,v}$ are outgoing from $m$ and have an empty intersection. We have $\varphi(m,L_{m,u}) = d(m,u)$, otherwise there would exist another vertex $u^*$, where $d(m,u^*) > d(m,u)$ and $L_{m,u} = L_{m,u^*}$, and we would obtain $d(u^*,v) = d(u^*,m) + d(m,v) > d(u,v)$, a contradiction. For the same reason, $\varphi(m,L_{m,v}) = d(m,v)$. So, we know that there is a pair $L_{m,u}\cap L_{m,v} = \emptyset$ of POFs outgoing from $m$ such that $d(u,v) = \varphi(m,L_{m,u}) + \varphi(m,L_{m,v})$. 

Now, suppose that there is another pair $L \cap L^* = \emptyset$ outgoing from $m$ such that $d(u,v) < \varphi(m,L) + \varphi(m,L^*)$. Theorem~\ref{th:disjoint_POF} implies that there are two vertices $x,y$ verifying $d(m,x) = \varphi(m,L)$, $d(m,y) = \varphi(m,L^*)$, and $m = m(x,y,v_0)$. So, $d(x,y) > d(u,v)$. As $m = m(x,y,v_0)$, we have a contradiction with the definition of $u,v$. In brief, $L_{m,u}, L_{m,v}$ is the pair of disjoint POFs outgoing from $m$ which maximizes the sum of labels. Equation~\eqref{eq:upsilon_m} holds.

Equation~\eqref{eq:upsilon_G} is a direct consequence of Equation~\eqref{eq:upsilon_m} because the diameter is the maximum of $\Upsilon(m)$ over all vertices $m \in V$ by definition.
\end{proof}

According to this corollary, a way to obtain the diameter is to find the triplet $(m,L,L^*)$ which maximizes $\varphi(m,L) + \varphi(m,L^*)$, where $L$ and $L^*$ are disjoint POFs outgoing from $m$. Our idea is to determine $\Upsilon(m)$ for each $m \in V$ and then to identifies the maximum of these values. We thus propose an algorithm returning all values $\Upsilon(m)$ in linear time for constant $d$.

\textbf{Computation of $\Upsilon(m)$.} We fix some $m \in V$. Let $\mathcal{L}_m$ be the set of POFs outgoing from $m$. Its cardinality is denoted by $N_m$. We know that $\mathcal{L}_m$ is closed under subsets and, for any $L \in \mathcal{L}_m$, $\card{L} \leq d$. A positive integer $\varphi(m,L)$ is associated with any of these sets. From now on, we denote them by $\omega(L)$ to be concise: $\omega(L) = \varphi(m,L)$. The goal is to find the pair $L,L^* \in \mathcal{L}_m$, $L \cap L^* = \emptyset$ maximizing $\omega(L) + \omega(L^*)$. We call this problem \textsc{maximum-weighted disjoint sets} (MWDS).

\begin{definition}[Maximum-weighted disjoint sets]~

\textbf{Input:} Elements $\mathcal{E}$, collection $\mathcal{L}_m$ of sets of elements, weights $\omega : \mathcal{L}_m \rightarrow \mathbb{N}^+$, parameter $d \in \mathbb{N}$. For each $L \in \mathcal{L}_m$, $\card{L}\le d$.

\textbf{Output:} Pair $L,L^* \in \mathcal{L}_m$, $L \cap L^* = \emptyset$ which maximizes $\omega(L) + \omega(L^*)$.
\end{definition}

We design an algorithm solving MWDS in time $O(2^{d\log d}N_m)$. Before describing it, we need some notation. Given $L \in \mathcal{L}_m$, we denote by $\opp(L)$ the \textit{opposite} of set $L$, {\em i.e.} the set in $\mathcal{L}_m$ verifying $L \cap \opp(L) = \emptyset$ with maximum weight $\omega(\opp(L))$. Our algorithm will consist in identifying the opposite $\opp(L)$ of each $L \in \mathcal{L}_m$.

We construct a tree $T_m$ which allows us to compute all opposites. To avoid confusions, the vertices of $T_m$ are called \textit{nodes}. Nodes of $T_m$ are indexed with POFs in $\mathcal{L}_m$ and edges of $T_m$ are indexed with $\Theta$-classes. For any node $a \in T_m$, we denote by $R(a)$ the union of indices of the edges which are on the simple path from the root of $T_m$ to node $a$. With our construction, we announce that any set $R(a)$ will be a POF.

We begin with an iterative presentation of our algorithm. We identify first the set $L_0 \in \mathcal{L}_m$ with the maximum weight. The running time of such step is $O(N_m)$. The root of our tree is indexed with $L_0$ (line~\ref{line:init_tree} of Algorithm~\ref{algo:tree}). At this step, we know that all sets $L\in \mathcal{L}_m$ verifying $L_0 \cap L = \emptyset$ admit $L_0$ as their opposite: $\opp(L) = L_0$. The remaining sets are the ones having a nonempty intersection with $L_0$.

\begin{figure}[h]
\centering
\begin{subfigure}[b]{0.54\columnwidth}
\centering
\scalebox{0.8}{\begin{tikzpicture}


\draw [color = black, fill = white] (6.0,10.0) -- (6.0,11.0) -- (9.0,11.0) -- (9.0,10.0) --  (6.0,10.0);
\draw [color = black, fill = white] (3.0,7.5) -- (3.0,8.5) -- (6.0,8.5) -- (6.0,7.5) -- (3.0,7.5);
\draw [color = black, fill = white] (1.0,5.0) -- (1.0,6.0) -- (4.0,6.0) -- (4.0,5.0) -- (1.0,5.0);

\node (P1) at (7.5,10.5) {$\set{E_i,E_j}$};
\node (P2) at (4.5,8.0) {$\set{E_j,E_h,E_r}$};
\node (P3) at (2.5,5.5) {$\set{E_{\ell}}$};



\node[scale=1.1, color = red] at (8.7,11.3) {$a_0$};
\node[scale=1.1, color = red] at (5.7,8.8) {$a_1$};
\node[scale=1.1, color = red] at (3.7,6.3) {$a_2$};

\node[scale=1.1] at (4.0,9.8) {$E_i$};
\node[scale=1.1] at (2.1,7.3) {$E_h$};

\node[scale=1.1] at (9.5,9.8) {$E_j$};

\node[scale=1.1] at (7.0,7.3) {$E_j$};


\draw[->,>=latex,rounded corners=5pt,line width = 1.4pt] (6.0,10.3) -| (3.5,8.5);
\draw[->,>=latex,rounded corners=5pt,line width = 1.4pt] (3.0,7.8) -| (1.5,6.0);

\draw[->,>=latex,rounded corners=5pt,line width = 1.4pt] (9.0,10.3) -| (10.0,8.5);

\draw[->,>=latex,rounded corners=5pt,line width = 1.4pt] (6.0,7.8) -| (7.5,6.0);

\end{tikzpicture}}
\caption{Tree $T_m$}
\label{subfig:tree}
\end{subfigure}
\begin{subfigure}[b]{0.44\columnwidth}
\centering
\scalebox{0.8}{\begin{tikzpicture}


\node[draw, circle, minimum height=0.2cm, minimum width=0.2cm, fill=black] (P31) at (5,1) {};
\node[draw, circle, minimum height=0.2cm, minimum width=0.2cm, fill=black] (P32) at (5,2.5) {};
\node[draw, circle, minimum height=0.2cm, minimum width=0.2cm, fill=black] (P33) at (6,1.4) {};
\node[draw, circle, minimum height=0.2cm, minimum width=0.2cm, fill=black] (P34) at (6,2.9) {};

\node[draw, circle, minimum height=0.2cm, minimum width=0.2cm, fill=black] (P41) at (7,1) {};
\node[draw, circle, minimum height=0.2cm, minimum width=0.2cm, fill=black] (P42) at (7,2.5) {};

\node[draw, circle, minimum height=0.2cm, minimum width=0.2cm, fill=black] (P51) at (9,1) {};
\node[draw, circle, minimum height=0.2cm, minimum width=0.2cm, fill=black] (P52) at (9,2.5) {};

\node[draw, circle, minimum height=0.2cm, minimum width=0.2cm, fill=black] (P61) at (8.0,1.4) {};
\node[draw, circle, minimum height=0.2cm, minimum width=0.2cm, fill=black] (P62) at (8.0,2.9) {};
\node[draw, circle, minimum height=0.2cm, minimum width=0.2cm, fill=black] (P63) at (10.0,1.4) {};
\node[draw, circle, minimum height=0.2cm, minimum width=0.2cm, fill=black] (P64) at (10.0,2.9) {};

\node[draw, circle, minimum height=0.2cm, minimum width=0.2cm, fill=black] (P71) at (7.0,4.0) {};
\node[draw, circle, minimum height=0.2cm, minimum width=0.2cm, fill=black] (P72) at (8.0,4.4) {};
\node[draw, circle, minimum height=0.2cm, minimum width=0.2cm, fill=black] (P73) at (9.0,4.0) {};
\node[draw, circle, minimum height=0.2cm, minimum width=0.2cm, fill=black] (P74) at (10.0,4.4) {};

\node[draw, circle, minimum height=0.2cm, minimum width=0.2cm, fill=black] (v0) at (4.7,4.2) {};


\draw[line width = 1pt] (P31) -- (P33);
\draw[line width = 1pt] (P32) -- (P34);
\draw[line width = 1pt] (P33) -- (P34);
\draw[line width = 1pt] (P61) -- (P33);
\draw[line width = 1pt] (P62) -- (P34);

\draw[line width = 1pt] (P31) -- (P32);
\draw[line width = 1pt] (P31) -- (P41);
\draw[<-,line width = 1.5pt, color = red] (P32) -- (P42);
\draw[<-,line width = 1.5pt, color = blue] (P41) -- (P42);

\draw[line width = 1pt](P41) -- (P51);
\draw[->,line width = 1.5pt, color = green] (P42) -- (P52);
\draw[line width = 1pt] (P51) -- (P52);

\draw[line width = 1pt] (P41) -- (P61);
\draw[->,line width = 1.5pt, color = orange] (P42) -- (P62);
\draw[line width = 1pt] (P51) -- (P63);
\draw[line width = 1pt] (P52) -- (P64);
\draw[line width = 1pt] (P61) -- (P62);
\draw[line width = 1pt] (P61) -- (P63);
\draw[line width = 1pt] (P62) -- (P64);
\draw[line width = 1pt] (P63) -- (P64);

\draw[line width = 1pt] (P62) -- (P72);
\draw[line width = 1pt] (P64) -- (P74);
\draw[line width = 1pt] (P72) -- (P74);

\draw[->,line width = 1.5pt, color = purple] (P42) -- (P71);
\draw[line width = 1pt] (P52) -- (P73);
\draw[line width = 1pt] (P71) -- (P73);
\draw[line width = 1pt] (P71) -- (P72);
\draw[line width = 1pt] (P73) -- (P74);

\draw[->,line width = 1pt, dotted] (v0) -- (P42);

\node[scale=1.2] at (6.7,2.2) {$m$};
\node[scale=1.2] at (5.1,4.3) {$v_0$};
\node[scale=1.2] at (9.5,0.7) {$m_{ij}$};
\node[scale=1.2] at (10.4,4.7) {$m_{jhr}$};
\node[scale=1.4, color = red] at (4.4,2.5) {$E_{\ell}$};
\node[scale=1.4, color = purple] at (7.0,4.6) {$E_{r}$};
\node[scale=1.4, color = blue] at (7.0,0.4) {$E_{i}$};
\node[scale=1.4, color = green] at (9.5,2.2) {$E_{j}$};
\node[scale=1.4, color = orange] at (8.4,3.3) {$E_{h}$};
\end{tikzpicture}}
\caption{Hypercubes with basis $m$}
\label{subfig:star}
\end{subfigure}

\caption{An example of tree $T_m$ which allows us to compute all opposites in $\mathcal{L}_m$.}
\label{fig:tree}
\end{figure}
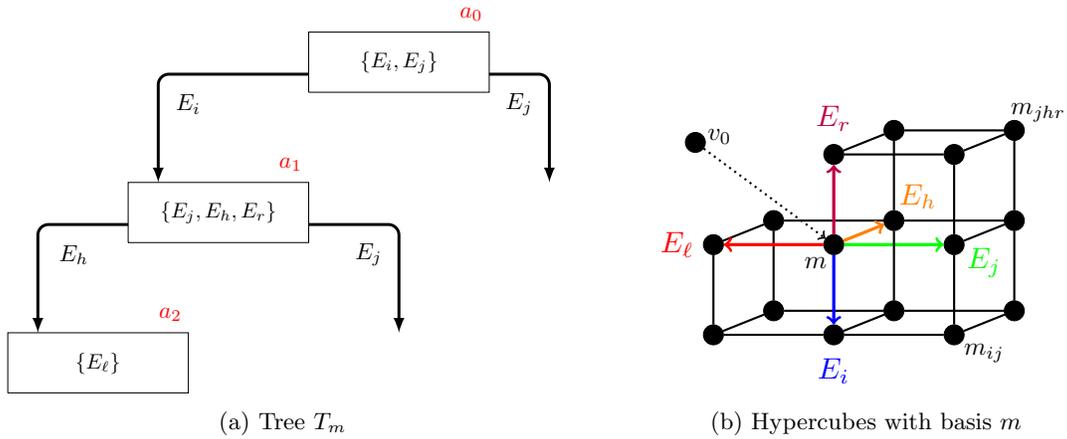

For each class $E_{i_0} \in L_0$, we add a child to $L_0$ and the edge of $T_m$ that connects the root with this child is indexed with $E_{i_0}$. For example, suppose that $L_0 = \set{E_i,E_j}$. Then, the root has two children. The two edges connecting $L_0$ to its children are respectively indexed with $E_i$ and $E_j$ (Figure~\ref{subfig:tree}). The idea is that if the intersection of a given set $L$ with $L_0$ is $\set{E_i}$ for example, then traverse the branch indexed with $E_i$ to find the opposite of $L$. If the intersection is $\set{E_i,E_j}$, then traverse arbitrarily one of the two branches $E_i$ and $E_j$.

The next step consists in indexing the children of $L_0$. We consider one of them being connected to $L_0$ by a branch with index $E_{i_0} \in L_0$. We denote this node by $a_1$. We compute the set $L_1 \in \mathcal{L}_m$ of maximum weight which does not contain any class of $R(a_1) = \set{E_{i_0}}$. We fix $L_1$ as the index of $a_1$. If we go back to our example, with $E_{i_0} = E_i$, then $L_1$ is the maximum-weighted POF in $\mathcal{L}_m$ that does not contain $E_i$. Suppose that $L_1 = \set{E_j,E_h,E_r}$. At this moment, we have treated the sets $L \in \mathcal{L}_m$ such that $L_0 \cap L = \emptyset$: their opposite is $L_0$. With this child $a_1$, we aim at finding the opposite of sets verifying $E_i \in L$. For example, for $L = \set{E_i,E_{\ell}}$, its opposite is indeed $\opp(L)=L_1 = \set{E_j,E_h,E_r}$ because $L \cap L_1 = \emptyset$ and $L_1$ is the maximum-weighted set which does not contain $E_i$. Unfortunately, if $L = \set{E_i,E_h}$, then $L$ and $L_1$ are not disjoint so we still have to find the opposite of $\set{E_i,E_h}$. Put formally, we know that $L_0$ is the opposite of any $L \in \mathcal{L}_m$ satisying $L \cap L_0 = \emptyset$. Moreover, we know that a child $L_1$ of $L_0$, connected via a branch $E_{i_0}\in L_0$, is the opposite of $L \in \mathcal{L}_m$ satisying $E_{i_0} \in L \cap L_0$ and $L \cap L_1 = \emptyset$.

We pursue the construction of the tree $T_m$. The node $a_1$ indexed by $L_1$ admits one child for each class $E_{i_1} \in L_1$ whose union with $R(a_1)$ form a POF. In our example, $R(a_1) = \set{E_i}$ and $L_1 = \set{E_j,E_h,E_r}$, so if for example $\set{E_i,E_h}$ is POF, then we add a child (say $a_2$) to $a_1$ and the index of the edge $(a_1,a_2)$ is $E_h$. The index of node $a_2$  is the POF of maximum weight which has an empty intersection with $R(a_2)$. In our example, $R(a_2) = \set{E_i,E_h}$, see Figure~\ref{subfig:tree}.

Algorithm~\ref{algo:tree} presents a recursive view of the construction of the tree. Symbol \# refers to the comments, outside the pseudocode. Figure~\ref{fig:tree} shows an example of tree $T_m$ and besides, an example of median graph where only the hypercubes with basis $m$ are represented. For each class, its edge adjacent to $m$ is colored and oriented in accordance with the $v_0$-orientation. As the root of $T_m$ is indexed with $\set{E_i,E_j}$, then we can affirm that the vertex $v$ satisfying $m \in I(v_0,v)$ which maximizes $d(m,v)$ verifies $L_{mv} = \set{E_i,E_j}$. Therefore, the longest shortest $(m,v)$-path, such that $m \in I(v_0,v)$, passes through vertex $m_{ij}$. As an example, the opposite of $\set{E_i}$ is $\set{E_j,E_h,E_r}$. So, if a diametral pair $(u,v)$ is such that $m = m(u,v,v_0)$ and $L_{mu} = \set{E_i}$, then the diameter is given by value $\varphi(m,\set{E_i}) + \varphi(m,\set{E_j,E_h,E_r})$ and a diametral path passes through vertex $m_{jhr}$.

\begin{algorithm}[t]
\SetKwFor{For}{for}{do}{\nl endfor}
\SetKwFor{Forall}{for all}{do}{\nl endfor}
\SetKwFor{Def}{def}{:}{\nl enddef}
\SetKwIF{If}{ElseIf}{Else}{if}{then}{else if}{else}{}
\DontPrintSemicolon
\SetNlSty{}{}{:}
\SetAlgoNlRelativeSize{0}
\SetNlSkip{1em}
\nl\KwIn{Vertex $m$, $\Theta$-classes $\mathcal{E}$, set $\mathcal{L}_m$ of POFs and $\omega : \mathcal{L}_m \rightarrow \mathbb{N}^+$}
\nl \KwOut{Tree $T_m$}
\nl \Def{{\em \textbf{children}}$(T,a)$}{
	\nl $L_a \leftarrow$ index of node $a$;\;
	\nl \Forall{$E_i \in L_a$}{
		\nl \If{$\set{E_i} \cup R(a)$ \mbox{\em is a POF}}{
			\nl add a child $a[E_i]$ of $a$ in $T$; Index $(a,a[E_i])$ with $E_i$; \label{line:add_child}\;
			\nl $L_{a[E_i]} \leftarrow \argmax\limits_{\substack{L \in \mathcal{L}_m \\ L \cap R(a[E_i]) = \emptyset}} \omega(L)$; \# here, $R(a[E_i]) = \set{E_i} \cup R(a)$\label{line:index_child}\; 
			\nl add index $L_{a[E_i]}$ to node $a[E_i]$;\;
		}
		\nl \ifend \;
	}
	\nl \Forall{children $a[E_i]$ of $a$ in $T$}{
		\nl \textbf{children}$(T,a[E_i])$;\label{line:recursion}\;
	}
}
\nl $L_0 \leftarrow \argmax_{L \in \mathcal{L}_m} \omega(L)$;\;
\nl $a_0 \leftarrow$ node with index $L_0$; $T_m \leftarrow$ single node $a_0$; \# here, $R(a) = \emptyset$\label{line:init_tree}\;
\nl \textbf{children}$(T_m,a_0)$;\;	

\caption{Construction of the tree $T_m$ providing us with all opposites.}
\label{algo:tree}
\end{algorithm}

Once the computation of the tree $T_m$ is completed, the index $L_a$ of node $a \in T_m$ is the opposite of any $L \in \mathcal{L}_m$ that contains all classes of $R(a)$ but no class of $L_a$.
Moreover, there is one child of node $a \in T_m$ for each class of its index $L_a$ whose union with $R(a)$ is a POF (line~\ref{line:add_child}). Indeed, if some $L \in \mathcal{L}_m$ contains $R(a)$ but also certain classes of the index $L_a$ of $a$, then we have to go down the tree through a branch indexed by an arbitrary class $E_i$ in $L_a \cap L$ to find the opposite of $L$. The following lemma ensures us that all opposites in $\mathcal{L}_m$ can be identified in the tree $T_m$.

\begin{lemma}
Let $a$ be a node of $T_m$, $L_a$ its index, and $L$ some POF in $\mathcal{L}_m$. Assume that $R(a) \subseteq L$. If $L_a \cap L = \emptyset$, then $\opp(L) = L_a$. Otherwise, there exists a child, denoted by $a[E_i]$, where the index $E_i$ of edge $(a,a[E_i])$ belongs to $L_a \cap L$.
\label{le:descent_tree}
\end{lemma}
\begin{proof}
By definition, the POF $L_a$ is the maximum-weighted POF which has no intersection with $R(a)$. If $L_a$ has no intersection with $L$, then we cannot find another POF disjoint from $L \supseteq R(a)$ with a greater weight than $\omega(L_a)$, so $\opp(L) = L_a$ in this case. Otherwise, if there is some $E_i \in L_a \cap L$, we can say that
$\set{E_i} \cup R(a) \subseteq L$ is a POF as its superset $L$. Consequently, as $E_i \in L_a$, there is an edge between $a$ and one of its children which has index $E_i$.
\end{proof}

We determine the opposite $\opp(L)$ of each $L \in \mathcal{L}_m$ with a tree search in $T_m$. If $L \cap L_0 = \emptyset$, then $\opp(L) = L_0$, else traverse the branch with an index in $L \cap L_0$, etc. Each node $a$ visited in this descent verifies $R(a) \subseteq L$. If $L_a \cap L = \emptyset$, then $\opp(L) = L_a$ (Lemma~\ref{le:descent_tree}). Otherwise, at least one child of $a$ can be visited to pursue the descent, via a branch $E_i \in L$ which exists according to Lemma~\ref{le:descent_tree}. The depth of $T_m$ is at most $d$ because $R(a)$ is a POF for any $a \in T_m$ and the cardinality of any POF is at most $d$. The tree search for each $L$ costs only $O(d^2)$ because we visit at most $d$ nodes (depth) and checking the intersections with all indices $L_a$ costs $O(d)$. The operation consisting in finding all opposites in $\mathcal{L}_m$ takes thus $O^*(N_m)$, where $O^*$ neglects polynomials of $d$. The computation of $\Upsilon(m)$ is complete: we pick up the pair $(L,\opp(L))$ with maximum weight. 

\textbf{Running time and discussion.} We provide an upper bound of the number of vertices of $T_m$. First, the depth of $T_m$ is at most $d$. Second, each node has at most $d$ children. Therefore, the number of nodes is at most $d!=O(2^{d\log d})$.
Computing the index of each of its nodes (line~\ref{line:index_child} of Algorithm~\ref{algo:tree}) takes $O(N_m)$. The construction of tree $T_m$ takes $O^*(2^{d\log d}N_m)$.

The running time for determining $\Upsilon(m)$ - construction of the tree and computation of the opposites - is $O^*(2^{d\log d}N_m)$.
Eventually, we pick up the vertex $m$ maximizing $\Upsilon(m)$ and the value obtained is $\diam(G)$ according to Corollary~\ref{co:diameter_sets}. The running time of the algorithm is thus $O^*(2^{d(\log d+1)}n)$ because $\sum_{m\in V} N_m = \alpha(G) \le 2^dn$.

Our algorithm outputs the diameter. However, we may aim at returning a diametral pair. To do so, we can simply add an extra label $\mu(u,L)$ in order to obtain a labeling pair $(\mu(u,L),\varphi(u,L))$, where $\mu(u,L)$ is a vertex which is at distance $\varphi(u,L)$ from $u \in I(v_0,\mu(u,L))$ with ladder set $L$. Two slight modifications must be done in Algorithm~\ref{algo:labels} presented in Sections~\ref{subsec:labels}. First, when $\varphi(u,L) = 0$ (line~\ref{line:phi_zero}), we initialize the vertex $\mu(u,L)$ with the anti-basis $u^+$. Second, the update of $\varphi(u^-,X)$ in line~\ref{line:update_phi} is replaced by an update of both the vertex $\mu(u^-,X)$ and the value $\varphi(u^-,X)$. If $\card{X} + \varphi(u,L) > \varphi(u^-,X)$, then $\mu(u^-,X) \leftarrow \mu(u,L)$.  

\section{Computing all eccentricities in linear time for dimension $d=O(1)$} \label{sec:eccentricities}

We use the computations of Section~\ref{sec:diameter} to determine the eccentricity of each vertex of $G$. We remind that, in time $O(2^{O(d\log d)}n)$, we obtained the following labels for any pair $(u,L)$, where $u\in V$ and $L$ is a POF outgoing from $u$.
\begin{itemize}
\item $\varphi(u,L)$: distance of the longest shortest path starting from vertex $u$ with ladder set $L$,
\item $\opp_u(L)$: ladder set $L^*$ outgoing from $u$ which has an empty intersection with $L$ and maximizes $\varphi(u,L^*)$.
\end{itemize}

We proceed in two steps in this section. First, we define the concept of \textit{milestones} which consist in particular vertices of any interval $I(u,v)$ satisfying $u \in I(v_0,v)$. Second, we compute new labels $\psi(u,X)$, where $u \in V$ and $X$ is a POF arriving in $u$. Their definition depend on the notion of milestone. Finally, we establish a relationship between the eccentricity $\ecc(u)$ of $u \in V$ and both labels $\varphi(u,L)$ and $\psi(u,X)$. 

\subsection{Milestones} \label{subsec:milestones}

We consider two vertices $u,v$ such that $u \in I(v_0,v)$, said differently $u = m(u,v,v_0)$. The notion of \textit{milestone} is defined recursively.

\begin{definition}[Milestones $\Pi(u,v)$]
Let $L_{u,v}$ be the ladder set of $u,v$ and $u^+$ be the anti-basis of the hypercube with basis $u$ and signature $L_{u,v}$.
If $u^+ = v$, then pair $u,v$ admits two milestones: $\Pi(u,v) = \set{u,v}$. Otherwise, the set $\Pi(u,v)$ is the union of $\Pi(u^+,v)$ with vertex $u$: $\Pi(u,v) = \set{u} \cup \Pi(u^+,v)$.
\label{def:milestones}
\end{definition}

\begin{figure}[h]
\centering
\scalebox{0.9}{\begin{tikzpicture}


\node[draw, circle, minimum height=0.2cm, minimum width=0.2cm, fill=black] (P11) at (1,1) {};
\node[draw, circle, minimum height=0.2cm, minimum width=0.2cm, fill=black] (P12) at (1,2.5) {};

\node[draw, circle, minimum height=0.2cm, minimum width=0.2cm, fill=red] (P21) at (3,1) {};
\node[draw, circle, minimum height=0.2cm, minimum width=0.2cm, fill=black] (P22) at (3,2.5) {};
\node[draw, circle, minimum height=0.2cm, minimum width=0.2cm, fill=black] (P23) at (3,4) {};

\node[draw, circle, minimum height=0.2cm, minimum width=0.2cm, fill=black] (P31) at (5,1) {};
\node[draw, circle, minimum height=0.2cm, minimum width=0.2cm, fill=red] (P32) at (5,2.5) {};
\node[draw, circle, minimum height=0.2cm, minimum width=0.2cm, fill=black] (P33) at (5,4) {};

\node[draw, circle, minimum height=0.2cm, minimum width=0.2cm, fill=black] (P41) at (7,1) {};
\node[draw, circle, minimum height=0.2cm, minimum width=0.2cm, fill=red] (P42) at (7,2.5) {};

\node[draw, circle, minimum height=0.2cm, minimum width=0.2cm, fill=black] (P51) at (9,1) {};
\node[draw, circle, minimum height=0.2cm, minimum width=0.2cm, fill=black] (P52) at (9,2.5) {};

\node[draw, circle, minimum height=0.2cm, minimum width=0.2cm, fill=black] (P61) at (8.0,1.4) {};
\node[draw, circle, minimum height=0.2cm, minimum width=0.2cm, fill=black] (P62) at (8.0,2.9) {};
\node[draw, circle, minimum height=0.2cm, minimum width=0.2cm, fill=black] (P63) at (10.0,1.4) {};
\node[draw, circle, minimum height=0.2cm, minimum width=0.2cm, fill=red] (P64) at (10.0,2.9) {};


\draw[line width = 1.4pt] (P11) -- (P12);
\draw[line width = 1.4pt] (P11) -- (P21);
\draw[line width = 1.4pt] (P12) -- (P22);
\draw[line width = 1.4pt,dashed,color = cyan] (P21) -- (P22);

\draw[line width = 1.4pt,dashed,color = cyan] (P21) -- (P31);
\draw[line width = 1.4pt,dashed,color = cyan] (P22) -- (P32);
\draw[line width = 1.4pt,dashed,color = cyan] (P31) -- (P32);

\draw[line width = 1.4pt] (P22) -- (P23);
\draw[line width = 1.4pt] (P23) -- (P33);
\draw[line width = 1.4pt] (P32) -- (P33);

\draw[line width = 1.4pt] (P31) -- (P41);
\draw[line width = 1.4pt, dashed, color = blue] (P32) -- (P42);
\draw[line width = 1.4pt] (P41) -- (P42);

\draw (P41) -- (P51);
\draw[line width = 1.4pt, dashed, color = purple] (P42) -- (P52);
\draw[line width = 1.4pt] (P51) -- (P52);

\draw (P41) -- (P61);
\draw[line width = 1.4pt, dashed, color = purple] (P42) -- (P62);
\draw (P51) -- (P63);
\draw[line width = 1.4pt, dashed, color = purple] (P52) -- (P64);
\draw[line width = 1.4pt] (P61) -- (P62);
\draw (P61) -- (P63);
\draw[line width = 1.4pt, dashed, color = purple] (P62) -- (P64);
\draw[line width = 1.4pt] (P63) -- (P64);


\node[scale=1.2, color = cyan] at (4.0,0.5) {$L_{u,v}$};
\node[scale=1.2, color = blue] at (6.0,2.9) {$L_{u^+,v}$};
\node[scale=1.2, color = purple] at (9.0,3.3) {$L_{u^{++},v}$};

\node[scale = 1.2] at (0.6,0.7) {$v_0$};
\node[scale = 1.2] at (2.6,0.7) {$u$};
\node[scale = 1.2] at (10.4,3.2) {$v$};

\node[scale = 1.2] at (4.6,2.2) {$u^+$};
\node[scale = 1.2] at (6.5,2.2) {$u^{++}$};

\end{tikzpicture}}
\caption{A pair $u,v$ with $u \in I(v_0,v)$ and its milestones $\Pi(u,v)$ in red.}
\label{fig:milestones}
\end{figure}
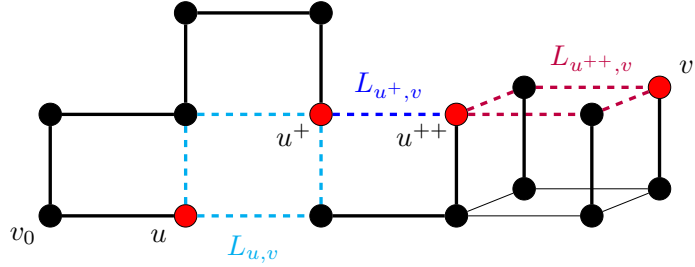

The milestones are the successive anti-bases of the hypercubes formed by the vertices and ladder sets traversed from $u$ to $v$. Moreover, it contains both $u$ and $v$. Concretely, the first milestone is $u$, the second is the anti-basis $u^+$ of the hypercube characterized by $(u,L_{u,v})$. The third one is the anti-basis $u^{++}$ of the hypercube with basis $u^+$ and $\Theta$-classes $L_{u^+,v}$, etc. All milestones are in the interval of $u$ and $v$: $\Pi(u,v) \subseteq I(u,v)$. 

Figure~\ref{fig:milestones} represents the same graph than in Figure~\ref{fig:ladder_sets} and shows the milestones $\Pi(u,v) = \set{u,u^+,u^{++},v}$. The hypercubes with the following pair basis-signature are highlighted with dashed edges: $(u,L_{u,v})$, $(u^+,L_{u^+,v})$, and $(u^{++},L_{u^{++},v})$.

We say that the milestone in $\Pi(u,v)$ which is different from $v$ but the closest to it is called the \textit{penultimate milestone}. We denote it by $\overline{\pi}(u,v)$. For example, $u^{++}$ is the penultimate milestone of $\Pi(u,v)$ in Figure~\ref{fig:milestones}. Furthermore, we denote by $\overline{L}_{u,v}$ the $\Theta$-classes of the hypercube with basis $\overline{\pi}(u,v)$ and anti-basis $v$. In Figure~\ref{fig:milestones}, $\overline{L}_{u,v} = L_{u^{++},v}$. 

A consequence of Theorem~\ref{th:pushing_label} is that, for two consecutive milestones in $\Pi(u,v)$, say $u$ and $u^+$ w.l.o.g, then $L_{u,v}$ and $L_{u^+,v}$ verify the following property: for any class $E_i \in L_{u^+,v}$, then $L_{u,v} \cup \set{E_i}$ is not a POF. We establish a property dealing with the penultimate milestone.

\begin{theorem}
Let $u,v \in V$ and $u \in I(v_0,v)$. Let $L$ be a POF outgoing from $v$ and $w$ the anti-basis of hypercube $(v,L)$. The following propositions are equivalent:
\begin{enumerate}
\item[(i)] vertex $v$ is the penultimate milestone of $(u,w)$: $\overline{\pi}(u,w) = v$,
\item[(ii)] the milestones of $(u,w)$ are the milestones of $(u,v)$ with $w$: $\Pi(u,w) = \Pi(u,v) \cup \set{w}$,
\item[(iii)] for any class $E_i \in L$, then $\overline{L}_{u,v} \cup \set{E_i}$ is not a POF.
\end{enumerate}
\label{th:penultimate}
\end{theorem}
\begin{proof}
\textit{Claim:} $\sigma_{u,w} = \sigma_{u,v} \cup L$. As $u \in I(v_0,v)$ and $v \in I(v_0,w)$, we have $u \in I(v_0,w)$ and $v \in I(u,w)$. The concatenation of a shortest $(u,v)$-path with a shortest $(v,w)$-path produces a shortest $(u,w)$-path. According to Theorem~\ref{th:signature}, $\sigma_{u,w} = \sigma_{u,v} \cup \sigma_{v,w}$. As $w$ is the anti-basis of hypercube $(v,L)$, then $\sigma_{v,w} = L$, which proves the claim.

\textit{(i)} $\Leftrightarrow$ \textit{(ii)}. Suppose that $v$ is the penultimate milestone in $\Pi(u,w)$. From the previous claim, we know that $\sigma_{u,w} = \sigma_{u,v} \cup L$. Consider a milestone in $\Pi(u,w)$ arriving before $v$, {\em i.e.} $p \in \Pi(u,w) \backslash \set{v,w}$, and suppose that some class $E_i \in L$ belongs to the ladder set $L_{p,w}$. The milestone after $p$ in $\Pi(u,w)$ is thus in $\partial H_i''$. By convexity of $\partial H_i''$, all milestones arriving after $p$ are in $\partial H_i''$. This is a contradiction as an edge of $E_i$ is outgoing from $v$, so $v \in \partial H_i'$. So, no class of $L$ appears in the ladder sets $L_{p,w}$ before milestone $v$. The milestones of pair $(u,w)$ are built with the classes in $\sigma_{u,v}$ until we arrive at $v$: $\Pi(u,w) = \Pi(u,v) \cup \set{w}$. Conversely, $(ii) \Rightarrow (i)$ is trivial.

\textit{(ii)} $\Leftrightarrow$ \textit{(iii)}.
The direct side is a consequence of Theorem~\ref{th:pushing_label}. As $\overline{\pi}(u,v) \in I(v_0,v)$ and $v \in I(v_0,w)$, we have $\overline{\pi}(u,v) \in I(v_0,w)$. Vertex $\overline{\pi}(u,v)$ belongs to $\Pi(u,w)$, so the ladder set of $\overline{\pi}(u,v)$ and $w$ is $\overline{L}_{u,v}$. We apply Theorem~\ref{th:pushing_label} to vertices $\overline{\pi}(u,v)$ and $w$: for any $E_i \in L_{v,w} = L$, set $\overline{L}_{u,v} \cup \set{E_i}$ is not a POF.

Suppose now that for any $E_i \in L$, set $\overline{L}_{u,v} \cup \set{E_i}$ is not a POF. 
The $\Theta$-classes in $L$ are outgoing from $v$. Assume a class $E_i$ of $L$ is outgoing from some milestone $p$ in $\Pi(u,v)$ different from $v$: $p \in \Pi(u,v)\backslash \set{v}$. By convexity of $\partial H_i'$, as $\set{p,v} \subseteq \partial H_i'$, we have $\overline{\pi}(u,v) \in \partial H_i'$. More generally, all vertices of the hypercube with basis $\overline{\pi}(u,v)$ and signature $\overline{L}_{u,v}$ are in $\partial H_i'$ for the same reason. Class $E_i$ form an isomorphism between $\partial H_i'$ and $\partial H_i''$: there is an hypercube with basis $\overline{\pi}(u,v)$ and classes $\overline{L}_{u,v} \cup \set{E_i}$. This is a contradiction as $\overline{L}_{u,v} \cup \set{E_i}$ is not a POF.

We just showed that no class of $L$ is adjacent to a milestone in $\Pi(u,v)\backslash \set{v}$. As $\sigma_{u,w} = \sigma_{u,v} \cup L$,  the milestones of $\Pi(u,w)$ are exactly the milestones of $\Pi(u,v)$ until we arrive at $v$. Indeed, the ladder set $L_{p,w}$ for any milestone $p \in \Pi(u,v)\backslash \set{v}$ cannot contain a class of $L$. At vertex $v$, the ladder set $L_{v,w}$ is $L$, so the next milestone is $w$ itself.
\end{proof}

This result is the keystone to compute the labels $\psi(u,X)$ and, finally, all eccentricities.

\subsection{Labels $\psi(u,X)$ and eccentricities} \label{subsec:labels_ecc}

Let $X$ be a POF ingoing to some vertex $u$ and $u^-$ be the basis of the hypercube with anti-basis $u$ and classes $X$. The label $\psi(u,X)$ is the maximum distance $d(u,v)$ we can obtain with a vertex $v$ satisfying the following properties:
\begin{itemize}
\item $m = m(u,v,v_0) \neq u$,
\item vertex $u^-$ is the penultimate milestone of pair $m,u$: $u^- = \overline{\pi}(m,u)$.
\end{itemize}

We present an inductive algorithm determining labels $\psi(u,X)$. 
We list the hypercubes of $G$ in the order given by list $\mathcal{Q}$ computed in Lemma~\ref{le:enum_hypercubes} and evoked in Algorithm~\ref{algo:labels}. Concretely, the hypercubes admitting $v_0$ as their basis are listed first. The ``peripheral'' hypercubes are listed last. For each of these hypercubes, we pick up their anti-basis $u$ and their signature $X$ and then we determine $\psi(u,X)$ as explained below.

Let $Q_0 \in \mathcal{Q}$ be an hypercube with basis $v_0$. This is the base case of our induction. We denote its anti-basis by $v_0^+$ and its classes by $X_0$. Figure~\ref{fig:psi_labels} presents an example where $X_0 = \set{E_i,E_j}$ and its opposite $\opp_{v_0}(X_0) = \set{E_{\ell},E_r}$. We aim at finding the maximum distance $d(v_0^+,v)$ with a vertex $v$ verifying $m=m(v_0^+,v,v_0) \neq v_0^+$ and $v_0 = \overline{\pi}(m,v_0^+)$. Necessarily, $m = v_0$ as $v_0 \in I(v_0^+,v)$. Therefore, we seek the longest shortest path starting from $v_0$ such that, concatenated with a shortest $(v_0,v_0^+)$ path, it produces a shortest path with $v_0^+$ as an extremity. As the distance between $v_0$ and $v_0^+$ is $\card{X_0}$, and according to Theorem~\ref{th:disjoint_POF}, we have:

\begin{equation}
\psi(v_0^+,X_0) = \card{X_0} + \varphi\left(v_0,\opp_{v_0}(X_0)\right)
\label{eq:psi_v_0}
\end{equation} 

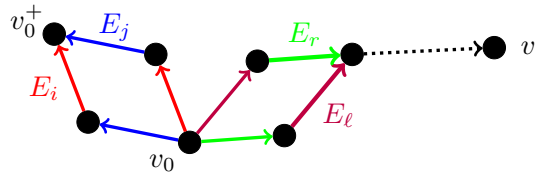
\begin{figure}[h]
\centering
\scalebox{0.9}{\begin{tikzpicture}


%
%

\node[draw, circle, minimum height=0.2cm, minimum width=0.2cm, fill=black] (P4) at (5.5,2.8) {};
\node[draw, circle, minimum height=0.2cm, minimum width=0.2cm, fill=black] (P42) at (5.0,4.1) {};

\node[draw, circle, minimum height=0.2cm, minimum width=0.2cm, fill=black] (P5) at (7,2.5) {};
\node[draw, circle, minimum height=0.2cm, minimum width=0.2cm, fill=black] (P52) at (6.5,3.8) {};

\node[draw, circle, minimum height=0.2cm, minimum width=0.2cm, fill=black] (P6) at (8.4,2.6) {};
\node[draw, circle, minimum height=0.2cm, minimum width=0.2cm, fill=black] (P62) at (8.0,3.7) {};

\node[draw, circle, minimum height=0.2cm, minimum width=0.2cm, fill=black] (P7) at (9.4,3.8) {};
\node[draw, circle, minimum height=0.2cm, minimum width=0.2cm, fill=black] (P8) at (11.5,3.9) {};



\draw[<-,line width = 1.4pt,color = blue] (P4) -- (P5);
\draw[->,line width = 1.4pt, color = green] (P5) -- (P6);
\draw[->,line width = 1.4pt, color = purple] (P5) -- (P62);
\draw[->,line width = 1.8pt, color = purple] (P6) -- (P7);
\draw[->,line width = 1.8pt, color = green] (P62) -- (P7);
\draw[->,line width = 1.4pt, dotted] (P7) -- (P8);

\draw[->,line width = 1.4pt, color = red] (P4) -- (P42);
\draw[->,line width = 1.4pt, color = red] (P5) -- (P52);
\draw[->,line width = 1.4pt, color = blue] (P52) -- (P42);


\node[scale=1.2, color = blue] at (5.9,4.2) {$E_j$};
\node[scale=1.2, color = red] at (4.85,3.3) {$E_i$};

\node[scale=1.2, color = green] at (8.7,4.1) {$E_r$};
\node[scale=1.2, color = purple] at (9.2,2.9) {$E_{\ell}$};

\node[scale = 1.2] at (6.6,2.2) {$v_0$};

\node[scale = 1.2] at (4.6,4.3) {$v_0^+$};
\node[scale = 1.2] at (12.0,3.9) {$v$};

\end{tikzpicture}}
\caption{The hypercubes with basis $v_0$ and signature $X_0 = \set{E_i,E_j}$ and $\opp_{v_0}(X_0) = \set{E_{\ell},E_r}$}
\label{fig:psi_labels}
\end{figure}

Let $Q \in \mathcal{Q}$ be an hypercube, with basis $u^-$, anti-basis $u$ and classes $X$. We suppose that the labels $\psi$ of vertex $u^-$ have already been computed. We are looking for the maximum distance $d(u,v)$ with a vertex $v$ verifying $m = m(u,v,v_0) \neq u$ and $u^- = \overline{\pi}(m,u)$. We distinguish two cases. 

First, we assume $m = u^-$. As for the base case (Equation~\eqref{eq:psi_v_0}), $\psi(u,X) = \card{X} + \varphi(u^-,\opp_{u^-}(X))$.

Second, we assume that $m \neq u^-$. Set $\Pi(m,u)$ admits at least three milestones: $m$, $u^-$, and $u$. Let $X^-$ be the POF ingoing to $u^-$ which is the ladder set in $\Pi(m,u)$ of the milestone just before $u^-$. According to Theorem~\ref{th:penultimate}, vertex $u^-$ is the penultimate milestone of $(m,u)$ if and only if $X^- \cup \set{E_i}$ is not a POF, for each $E_i \in X$. For this reason, value $\psi(u,X)$ can be expressed as:

\[
\psi(u,X) = \max\limits_{\substack{X^- ~\mbox{\scriptsize{POF ingoing to}}~ u^- \\ \forall E_i \in X, X^- \cup \set{E_i} ~\mbox{\scriptsize{not POF}}}} \card{X} + \psi(u^-,X^-)
\]

Our algorithm consists in taking the maximum value between the two cases. In this way, we obtain all labels $\psi(u,X)$ in time $O^*(2^d\alpha(G)) = O^*(2^{2d}n)$ because we consider each hypercube of the graph and list the POFs ingoing to their basis.

Eventually, we show how to compute all eccentricities in function of labels $\varphi(u,L)$ and $\psi(u,X)$.

\begin{theorem}
The eccentricity of vertex $u \in V$ is the maximum between all values $\varphi(u,L)$ and also all values $\psi(u,X)$. Formally,
\label{th:eccentricity}
\end{theorem}
\[
\ecc(u) = \mbox{max}\set{\max\limits_{\substack{L ~\mbox{\scriptsize{POF}} \\ \mbox{\scriptsize{outgoing from}}~ u}} \varphi(u,L), \max\limits_{\substack{X ~\mbox{\scriptsize{POF}} \\ \mbox{\scriptsize{ingoing to}}~ u}} \psi(u,X)}
\]
\begin{proof}
Let $v$ be a vertex in $G$ such that $\ecc(u) = d(u,v)$. If $m = m(u,v,v_0) = u$, then $u \in I(v_0,v)$ and value $d(u,v)$ is given by the label $\varphi(u,L)$, where $L$ is the ladder set $L_{u,v}$. Conversely, each $\varphi(u,L')$ is the distance between $u$ and some vertex $v'$ satisfying $u \in I(v_0,v')$. Therefore, the maximum value of $\varphi(u,L)$ is the maximum distance $d(u,v)$ such that $m = u$.

If $m \neq u$, let $u^-$ be the penultimate milestone in $\Pi(m,u)$ and $X$ be the classes of the hypercube with basis $u^-$ and anti-basis $u$. In this case, distance $d(u,v)$ is given by label $\psi(u,X)$. So, the maximum value $\psi(u,X)$ is the maximum distance $d(u,v)$ such that $m \neq u$.
\end{proof}

There are as many labels $\varphi(u,L)$ as hypercubes and the number of labels $\psi(u,X)$ is upper-bounded by $2^dn$. In brief, computing all eccentricities takes $2^{O(d\log d)}n$: the most expensive operation is still Algorithm~\ref{algo:tree} which determines the opposites of any pair $(u,L)$. 

\section{Conclusion} \label{sec:conclusion}

The main contribution of this article is an algorithm determining the diameter, the radius, and all eccentricities in median graphs in linear time when $d = O(1)$. A natural question is whether the techniques we used can be extended to propose an algorithm with a subquadratic running time for all median graphs. 

We believe some of the notions we proposed characterize shortest paths in median graphs and are powerful tools to determine the diameter and, more generally, all eccentricities. Ladder sets, ladder pairs but also milestones allow us to focus on certain distances instead of considering all pairs of vertices, as with a multiple BFS. The dependence of distances between two ``adjacent'' ladder sets (Theorem~\ref{th:pushing_label}) and the characterization of the diameter via ladder pairs (Corollary~\ref{co:diameter_sets}) are results which are likely to be used again in future research for median graphs.

The design of a linear-time algorithm for the diameter on median graphs seems compromised. Indeed, in~\cite{DuHaVi20}, the authors put in evidence the difficulty to design a linear-time algorithm computing the diameter for graphs with unbounded distance VC-dimension, which is at least $d$ on median graphs. Nevertheless, in our opinion, the design of a subquadratic algorithm is a challenging line of research.

\bibliographystyle{plain}
\bibliography{median}

\begin{thebibliography}{10}

\bibitem{AbGrWi15}
A.~Abboud, F.~Grandoni, and V.~V. Williams.
\newblock Subcubic equivalences between graph centrality problems, {APSP} and
  diameter.
\newblock In {\em Proc. of {SODA}}, pages 1681--1697, 2015.

\bibitem{AbWiWa16}
A.~Abboud, V.~V. Williams, and J.~R. Wang.
\newblock Approximation and fixed parameter subquadratic algorithms for radius
  and diameter in sparse graphs.
\newblock In {\em Proc. of {SODA}}, pages 377--391, 2016.

\bibitem{AlYuZw97}
N.~Alon, R.~Yuster, and U.~Zwick.
\newblock Finding and counting given length cycles.
\newblock {\em Algorithmica}, 17(3):209--223, 1997.

\bibitem{Av61}
S.~P. Avann.
\newblock Metric ternary distributive semi-lattices.
\newblock {\em Proc. Amer. Math. Soc.}, 12:407--414, 1961.

\bibitem{BaBrChKlKoSu10}
K.~Balakrishnan, B.~Bresar, M.~Changat, S.~Klavzar, M.~Kovse, and A.~R.
  Subhamathi.
\newblock Computing median and antimedian sets in median graphs.
\newblock {\em Algorithmica}, 57(2):207--216, 2010.

\bibitem{Ba84}
H.~Bandelt.
\newblock Retracts of hypercubes.
\newblock {\em Journal of Graph Theory}, 8(4):501--510, 1984.

\bibitem{BaBa84}
H.~Bandelt and J.~Barth{\'{e}}lemy.
\newblock Medians in median graphs.
\newblock {\em Discret. Appl. Math.}, 8(2):131--142, 1984.

\bibitem{BaCh08}
H.~Bandelt and V.~Chepoi.
\newblock Metric graph theory and geometry: a survey.
\newblock {\em Contemp. Math.}, 453:49--86, 2008.

\bibitem{BaChDrKo06}
H.~Bandelt, V.~Chepoi, A.~W.~M. Dress, and J.~H. Koolen.
\newblock Combinatorics of lopsided sets.
\newblock {\em Eur. J. Comb.}, 27(5):669--689, 2006.

\bibitem{BaChEp10}
H.~Bandelt, V.~Chepoi, and D.~Eppstein.
\newblock Combinatorics and geometry of finite and infinite squaregraphs.
\newblock {\em {SIAM} J. Discret. Math.}, 24(4):1399--1440, 2010.

\bibitem{BaQuSaMa02}
H.~Bandelt, L.~Quintana-Murci, A.~Salas, and V.~Macaulay.
\newblock The fingerprint of phantom mutations in mitochondrial dna data.
\newblock {\em Am. J. Hum. Genet.}, 71:1150--1160, 2002.

\bibitem{BaFoSyRi95}
H.~J. Bandelt, P.~Forster, B.~C. Sykes, and M.~B. Richards.
\newblock Mitochondrial portraits of human populations using median networks.
\newblock {\em Genetics}, 141(2):743--753, 1995.

\bibitem{BaCo93}
J.~Barth{\'{e}}lemy and J.~Constantin.
\newblock Median graphs, parallelism and posets.
\newblock {\em Discret. Math.}, 111(1-3):49--63, 1993.

\bibitem{BeChChVa20}
L.~B{\'{e}}n{\'{e}}teau, J.~Chalopin, V.~Chepoi, and Y.~Vax{\`{e}}s.
\newblock Medians in median graphs and their cube complexes in linear time.
\newblock In {\em Proc. of {ICALP}}, volume 168, pages 10:1--10:17, 2020.

\bibitem{BiKi47}
G.~Birkhoff and S.~A. Kiss.
\newblock A ternary operation in distributive lattices.
\newblock {\em Bull. Amer. Math. Soc.}, 53:745--752, 1947.

\bibitem{BoCrHaKoMaTa15}
M.~Borassi, P.~Crescenzi, M.~Habib, W.~A. Kosters, A.~Marino, and F.~W. Takes.
\newblock {Fast diameter and radius BFS-based computation in (weakly connected)
  real-world graphs: With an application to the six degrees of separation
  games}.
\newblock {\em Theor. Comput. Sci.}, 586:59--80, 2015.

\bibitem{BrKlSk06}
B.~Bresar, S.~Klavzar, and R.~Skrekovski.
\newblock Roots of cube polynomials of median graphs.
\newblock {\em Journal of Graph Theory}, 52(1):37--50, 2006.

\bibitem{Ca17}
S.~Cabello.
\newblock Subquadratic algorithms for the diameter and the sum of pairwise
  distances in planar graphs.
\newblock In {\em Proc. of {SODA}}, pages 2143--2152, 2017.

\bibitem{ChLaRoScTaWi14}
S.~Chechik, D.~H. Larkin, L.~Roditty, G.~Schoenebeck, R.~E. Tarjan, and V.~V.
  Williams.
\newblock Better approximation algorithms for the graph diameter.
\newblock In {\em Proc. of {SODA}}, pages 1041--1052, 2014.

\bibitem{Ch00}
V.~Chepoi.
\newblock Graphs of some {CAT(0)} complexes.
\newblock {\em Adv. Appl. Math.}, 24(2):125--179, 2000.

\bibitem{ChDrVa02}
V.~Chepoi, F.~F. Dragan, and Y.~Vax{\`{e}}s.
\newblock Center and diameter problems in plane triangulations and
  quadrangulations.
\newblock In {\em Proc. of {SODA}}, pages 346--355, 2002.

\bibitem{ChLaRa19}
V.~Chepoi, A.~Labourel, and S.~Ratel.
\newblock Distance labeling schemes for cube-free median graphs.
\newblock In {\em Proc. of MFCS}, volume 138, pages 15:1--15:14, 2019.

\bibitem{CoDrKo03}
D.~G. Corneil, F.~F. Dragan, and E.~K{\"{o}}hler.
\newblock On the power of {BFS} to determine a graph's diameter.
\newblock {\em Networks}, 42(4):209--222, 2003.

\bibitem{Du20}
G.~Ducoffe.
\newblock Isometric embeddings in trees and their use in the diameter problem.
\newblock {\em CoRR}, abs/2010.15803, 2020.

\bibitem{DuHaVi20}
G.~Ducoffe, M.~Habib, and L.~Viennot.
\newblock Diameter computation on \emph{H}-minor free graphs and graphs of
  bounded (distance) {VC-dimension}.
\newblock In {\em Proc. of {SODA}}, pages 1905--1922, 2020.

\bibitem{HaImKl99}
J.~Hagauer, W.~Imrich, and S.~Klavzar.
\newblock Recognizing median graphs in subquadratic time.
\newblock {\em Theor. Comput. Sci.}, 215(1-2):123--136, 1999.

\bibitem{HaImKl11}
R.~Hammack, W.~Imrich, and S.~Klavzar.
\newblock {\em Handbook of Product Graphs, Second Edition}.
\newblock CRC Press, Inc., 2011.

\bibitem{Ha73}
G.~Handler.
\newblock Minimax location of a facility in an undirected tree graph.
\newblock {\em Transp. Sci.}, 7:287--293, 1973.

\bibitem{ImKlMu99}
W.~Imrich, S.~Klavzar, and H.~M. Mulder.
\newblock Median graphs and triangle-free graphs.
\newblock {\em {SIAM} J. Discret. Math.}, 12(1):111--118, 1999.

\bibitem{KlMu99}
S.~Klavzar and H.~M. Mulder.
\newblock Median graphs: Characterizations, location theory and related
  structures.
\newblock {\em J. Combin. Math. Combin. Comput.}, 30:103--127, 1999.

\bibitem{KlMuSk98}
S.~Klavzar, H.~M. Mulder, and R.~Skrekovski.
\newblock An {Euler}-type formula for median graphs.
\newblock {\em Discret. Math.}, 187(1-3):255--258, 1998.

\bibitem{Ko09}
M.~Kovse.
\newblock Complexity of phylogenetic networks: counting cubes in median graphs
  and related problems.
\newblock {\em Analysis of complex networks: From Biology to Linguistics},
  pages 323--350, 2009.

\bibitem{MaLaHa08}
C.~Magnien, M.~Latapy, and M.~Habib.
\newblock Fast computation of empirically tight bounds for the diameter of
  massive graphs.
\newblock {\em {ACM} J. Exp. Algorithmics}, 13, 2008.

\bibitem{MoMuRo98}
F.~R. McMorris, H.~M. Mulder, and F.~S. Roberts.
\newblock The median procedure on median graphs.
\newblock {\em Discret. Appl. Math.}, 84(1-3):165--181, 1998.

\bibitem{MuSc79}
H.~M. Mulder and A.~Schrijver.
\newblock Median graphs and {Helly} hypergraphs.
\newblock {\em Discret. Math.}, 25(1):41--50, 1979.

\bibitem{Mu78}
M.~Mulder.
\newblock The structure of median graphs.
\newblock {\em Discret. Math.}, 24(2):197--204, 1978.

\bibitem{Mu80}
M.~Mulder.
\newblock The interval function of a graph.
\newblock {\em Mathematical Centre Tracts, Mathematisch Centrum, Amsterdam},
  1980.

\bibitem{RoWi13}
L.~Roditty and V.~V. Williams.
\newblock Fast approximation algorithms for the diameter and radius of sparse
  graphs.
\newblock In {\em Proc. of {STOC}}, pages 515--524, 2013.

\bibitem{SaNiWi93}
V.~Sassone, M.~Nielsen, and G.~Winskel.
\newblock A classification of models for concurrency.
\newblock In {\em Proc. of {CONCUR}}, volume 715 of {\em Lecture Notes in
  Computer Science}, pages 82--96, 1993.

\end{thebibliography}

\end{document}